\documentclass[12pt]{article}

\usepackage{amsmath,epsfig}
\usepackage{amssymb}
\usepackage{color}
\usepackage[english]{babel}
\usepackage{graphicx}
\usepackage{times}
\usepackage{verbatim}
\usepackage{psfrag}
\usepackage{multirow}

\newtheorem{thm}{Theorem}[section]
\newtheorem{lem}[thm]{Lemma}
\newtheorem{defn}[thm]{Definition}
\newtheorem{cor}[thm]{Corollary}
\newtheorem{prop}[thm]{Proposition}

\newtheorem{ex}[thm]{Example}

\newenvironment{proof}{\noindent\textsc{Proof: }}{\hfill$\fbox{}$\par\medskip\par}

\newenvironment{aenum}{\begin{enumerate}
 
 }{\end{enumerate}}

\newcommand{\F}{{\mathbb F}}

\newcommand{\R}{{\mathbb R}}
\newcommand{\Z}{{\mathbb Z}}
\newcommand{\N}{{\mathbb N}}

\newcommand{\cK}{{\cal K}}
\newcommand{\cL}{{\cal L}}

\newcommand{\cV}{{\cal V}}

\newcommand{\diam}{\mbox{\rm diam\,}}

\newcommand{\rank}{\mbox{\rm rank\,}}

\def\mapright#1{\stackrel{#1}{\longrightarrow}}

\def\mapdown#1{\Big\downarrow\rlap{$\vcenter{\hbox{$\scriptstyle#1$}}$}}

\newcommand{\aphi}{{\varphi^\urcorner}} 
\newcommand{\apsi}{{\psi^\urcorner}} 

\newcommand{\lphi}{{\overline{\varphi}}} 

\newcommand{\mphi}{{\tilde{\varphi}}} 
\newcommand{\mpsi}{{\tilde{\psi}}} 

\newcommand{\dm}{\mathrm{D}\,} 
\newcommand{\dmm}{\mathrm{D_m}\,} 

\title{Comparison of Persistent Homologies for Vector Functions: from continuous to discrete and back}
\author{Niccol\`o Cavazza, Marc Ethier, Patrizio Frosini,\\
Tomasz Kaczynski, and Claudia Landi}

\begin{document}

\maketitle

\begin{abstract}

The theory of multidimensional persistent homology was initially developed in the discrete setting, and involved the study of simplicial complexes filtered through an ordering of the simplices. Later, stability properties of multidimensional persistence have been proved to hold when topological spaces are filtered by continuous functions, i.e. for continuous data. This paper aims to provide a bridge between the continuous setting, where stability properties hold, and the discrete setting, where actual computations are carried out. More precisely, a stability preserving method is developed to compare rank invariants of vector functions obtained from discrete data. These advances confirm that multidimensional persistent homology is an appropriate tool for shape comparison in computer vision and computer graphics applications. The results are supported by numerical tests.
\end{abstract}



\section{Introduction}

In this paper we present a discrete counterpart of the theory of persistent homology of vector functions that still guarantees stability properties as the continuous framework. The theory of multidimensional persistence was developed in the discrete setting in \cite{CaZo07}, and involved the study of simplicial complexes filtered through an ordering of the simplices. On the other hand stability properties of multidimensional persistence are proved to hold when triangulable spaces are filtered by continuous functions, i.e. for continuous data \cite{FrMu99,CeDi*10}. This paper aims to be a bridge between the continuous setting, where stability properties hold, and the discrete setting, where actual computations are carried out.
More precisely, we develop a method to compare persistent homologies of vector functions obtained from discrete data. We show that in the passage from the continuous to the discrete framework stability is preserved. These advances support the appropriateness of multidimensional persistent homology for shape comparison by functions.

The problem of comparing shapes is well-studied in computer vision and computer graphics and many algorithms have been developed for this purpose. A widely used scheme is to associate a shape with a shape descriptor, or a signature, and comparing shapes by measuring dissimilarity between descriptors.
An important class of shape descriptors, which may be called {\em shape-from-functions} methods, is based on the common idea of performing a topological exploration of the shape according to some quantitative geometric properties provided
by a (measuring) function defined on the shape and chosen to extract shape features \cite{BiDe*08}.

The simplest topological attribute of a space is the number of its connected components. A well-known mathematical tool to count the number of connected components is the homology group $H_0$. More complex topological features are revealed by higher homology groups.

Persistent homology is a shape-from-functions method for shape description involving homology groups of any degree. The idea is to filter a space by the sublevel sets of the function and to analyze the homological changes of the sublevel sets across this filtration, due to the appearance or disappearance of topological attributes, such as connected components. Features with a short persistence along the filtration can be regarded as negligible information due to noise or very fine details. For application purposes, it is often sufficient to disregard the group structure of persistent homology and retain only the rank information. This gives rise to the notions of {\em rank invariant} \cite{CaZo07}, persistent Betti numbers \cite{EdLeZo02}, size functions \cite{VeUr*93}.

The topic has been widely studied in the case of filtrations induced by scalar continuous functions (i.e. one-dimensional persistence), especially in connection with the stability problem \cite{CoEdHa07,ChCo*09,CoEdHaMi10, dAFrLa}.

This theory has been generalized to a multidimensional situation in which a vector-valued function characterizes
the data as suggested in \cite{EdHa02,EdHa08}. Results in this area are given in \cite{BiCe*08,CaZo07, CaDiFe10,CeDi*10}. This generalization is quite natural in view of the analogous generalization of Morse Theory \cite{Sma73}. Moreover, it is motivated by applications where data are more completely described by more than one function (e.g., curvature and torsion for space curves).

The passage from scalar to vector-valued functions presents new challenges. To begin with, critical points are no longer isolated even in non-degenerate situations \cite{EdHa02}. Although the relevant points for persistent homology of vector functions are a subset of the critical points, precisely the Pareto critical points, these are still non-isolated \cite{CeFr10}. For example, in the case of the sphere $x^2+y^2+z^2=1$ with the function $f=(y,z)$, the Pareto critical points are those in the set $ x=0$, $y^2+z^2=1$, $yz\ge 0$.

Another delicate issue is passing from the comparison of continuous models to that of discrete models. This is an essential passage, and the core of this paper. Indeed, for two given real-world objects $X$ and $Y$, modeled as triangulable topological spaces (e.g., manifolds), we usually only know simplicial descriptions $\cK$ and $\cL$ of them, affected by approximation errors. For example, acquiring 3D models of real-world objects for computer graphics applications needs to account for errors due to sensor resolution, noise in the measurements, inaccuracy of sensor calibration \cite{BeRu02}. Moreover, different techniques for reconstructing the geometry and topology of the scanned object yield different polyhedral approximations. Analogous considerations hold for any continuous measuring functions $f:X\to \R^k$ and $g:Y\to \R^k$, because we could only consider approximations $\mphi:K\to \R^k$, $\mpsi:L\to\R^k$ of $f$ and $g$ defined on finite polyhedra. Depending on the context of a specific application, these functions may or may not be given by explicit formulas. In either case, it is legitimate to assume that we are able to compute their values on vertices of $\cK$ and $\cL$. Hence, we only know the discrete maps $\varphi:\cV(\cK)\to \R^k$, $\psi :\cV(\cL)\to \R^k$ which are the restrictions of $\mphi$ and $\mpsi$ to vertices. Therefore a natural question is whether shape comparison by persistent homology of vector functions is numerically stable, i.e. whether the computation of a distance between rank invariants of discrete models gives a good approximation of the ideal distance between rank invariants of continuous models.

Our main result, Theorem \ref{thm:main}, gives an affirmative answer to this question. It states that, in the passage from continuous to discrete data, the distance between rank invariants does not increase, provided that stability holds for the continuous model. We underline that at least one stable distance between rank invariants of continuous vector functions exists as proved in \cite{CeDi*10}.
In order to profit from the stability theory in the continuous case, we give a new construction of axis-wise linear interpolation $\aphi$ which is generic in the sense that its persistent homology is exactly equal to that of the map $\varphi$ defined on vertices. 
In addition, this axis-wise interpolation can be used with stable distances to obtain a measure of how much a model can be simplified in order to ease the computation of shape signatures. Indeed, the computation time can become prohibitive when using large simplicial complexes to represent models, which is why using coarser representations can become necessary. Since doing so comes at a cost in terms of accuracy, we can, given an allowed error threshold, determine the level of precision required to respect this threshold.

The paper is organized as follows. In Section \ref{sec:background} the necessary background notions concerning persistence are reviewed and put in the context of our aims. Section \ref{sec:simplicial} starts with the description of the simplicial framework and with Example~\ref{ex:tetrahedron} which is a simplicial analogue of the sphere example pointed above. The same example shows that, in the vector case, the linear extension of a map defined on vertices does not satisfy the genericity property described above. Topological artifacts of an interpolation method have been observed before. This phenomenon can be referred to as {\em topological aliasing}. Our example motivates the construction of our axis-wise linear interpolation. We next prove Theorem \ref{th:sublevel-cont-vs-simplicial} on the deformation retraction of continuous sublevel sets of $\aphi$ onto the simplicial sublevel sets of $\varphi$. We introduce the notion of homological critical value for vector functions. As in the sphere example, the set of critical values need not be discrete, but we prove in Theorem \ref{th:hom-crit-C} that in the case of $\aphi$ it has to be contained in a finite union of hyperplanes, thus it is a nowhere dense set and its $k$-dimensional Lebesgue measure is zero (Corollary \ref{cor:nowhere-dense}).

Section \ref{Sec:approximation} starts with Lemma \ref{lem:approx} that provides an approximation of a distance between the rank invariants of continuous functions by that of the rank invariants of the corresponding axis-wise linear approximations. The genericity of $\aphi$ allows us to introduce the rank invariant for $\varphi$. Although this rank invariant is defined for a discrete function $\varphi$ and computed using only simplicial sublevel sets, it takes pairs of real vectors as variables, as it is in the case of the rank invariant for continuous functions. We show that this new rank invariant for $\varphi$ is equal to that of $\aphi$. This allows us to derive the main result of the paper (Theorem~\ref{thm:main}).

Section \ref{sec:experimentation} describes an algorithm which computes an approximate matching distance. Our algorithm is a modification of the algorithm described in \cite{BiCe*10}, adapted to the rank invariants. The correctness of the algorithm is guaranteed by the results of Section \ref{Sec:approximation}. We next present tests of the algorithm performed on simplicial models in the case $k=2$. Our tests revealed the same discrepancy as observed in Example \ref{ex:tetrahedron}, thus providing numerical confirmation of topological aliasing. 
Finally, as a practical implication of our theoretical results, we present a procedure to predetermine to which extent data resolution can be coarsened in order to maintain a certain error threshold on rank invariants.

\section{Basic notions and working assumptions}
\label{sec:background}

Let us consider a triangulable topological space $X$ (i.e., a space homeomorphic to the carrier of a finite simplicial complex). A {\em filtration} of $X$ is a family $\mathcal{F}=\{X_\alpha\}_{\alpha\in \R^k}$ of subsets of $X$ that are nested with respect to inclusions, that is: $X_\alpha\subseteq X_{\beta}$, for every $\alpha\preceq \beta$, where $\alpha\preceq \beta$ if and only if $\alpha_j \leq \beta_j$ for all $j=1,2,\ldots, k$.

Persistence is based on analyzing the homological changes occurring along the filtration as $\alpha$ varies. This analysis is carried out by considering, for $\alpha\preceq\beta$, the homomorphism
\[
H_*(i^{(\alpha,\beta)}): H_*(X_{\alpha}) \to
H_*(X_{\beta}).
\]
induced by the inclusion map $i^{(\alpha,\beta)}:X_{\alpha}\hookrightarrow X_{\beta}$.
We work with \v{C}ech homology with coefficients in a given field $\F$. When each $X_\alpha$, $\alpha\in\R^k$, is triangulable, it reduces to simplicial homology.
For simplicity of notation we write $H_*(X_\alpha)$ for the graded homology space $H_*(X_\alpha;\F)=\{H_q(X_\alpha;\F)\}_{q\in\Z}$. The choice of a field is only made in experimentations, the most convenient in computations being $\F=\Z_p$, with $p$ a prime number. Thus, for any $q\in\Z$, $H_q(X_\alpha)$ is a vector space of dimension equal to the $q$'th Betti number of $X_\alpha$.

The image of the map $H_q(i^{(\alpha,\beta)})$ is a vector space known as the {\em $q$'th persistent homology group} of the filtration at $(\alpha,\beta)$. It contains the homology classes of order $q$ born not later than $\alpha$ and still alive at $\beta$. The dimension of this vector space is called a {\em $q$'th persistent Betti number}.

A rank invariant is a function that encodes the changes in the persistent Betti numbers as $\alpha$ and $\beta$ vary.
Setting
\[
\Delta^k_+:=\{(\alpha,\beta)\in \R^k\times \R^k \mid \alpha \prec \beta\},
\]
where $\alpha\prec \beta$ if and only if $\alpha_j < \beta_j$ for all $j=1,2,\ldots, k$,
the {\em $q$'th rank invariant} of the filtration $\mathcal F$ is the function $\rho_{\mathcal F}^q:\Delta^k_+ \to \N\cup\{\infty\}$ defined on each pair $(\alpha,\beta)\in \Delta^k_+$ as the rank of the map $H_q(i^{(\alpha,\beta)})$. In other words, $\rho_{\mathcal F}^q(\alpha,\beta)=\dim \mathrm{im} H_q(i^{(\alpha,\beta)})$.

In this paper, we will use the notation $\rho_{\mathcal F}$ to refer to rank invariants of arbitrary order. Ultimately, the shapes of two triangulable spaces $X$ and $Y$, filtered by $\mathcal F$ and $\mathcal G$, respectively, can be compared by using an (extended) distance $\dm$ between their rank invariants
$\rho_{\mathcal F}$ and $\rho_{\mathcal G}$.

The framework described so far for general filtrations can be specialized in various directions. We now review the two most relevant ones for our paper.

\subsection{Persistence of sublevel set filtrations}
\label{sec:stability-property}

Given a continuous function $f:X\to \R^k$, it induces on $X$ the so-called {\em sublevel set filtration}, defined as follows:
\[
X_\alpha=  \{x\in X \mid f(x)\preceq \alpha\}.
\]

We will call the function $f$ a {\em measuring function} and denote the rank invariant associated with this filtration by $\rho_f$.

Since $X$ is assumed to be triangulable and $f$ is continuous, $\rho_f(\alpha,\beta)<+\infty$ for every $\alpha\prec \beta\in\R^k$ (cf. \cite{CaLa11}).

Among all the (extended) distances $\dm$ between rank invariants of filtrations, we confine our study to those ones that, when applied to sublevel set filtrations, satisfy the following {\em stability property}:

\begin{enumerate}
\item[(S)] For every $f,f':X\to \R^k$ continuous functions, $\dm(\rho_f,\rho_{f'})\le \|f-f'\|_\infty$
where $\|f\|_\infty = \max_{x\in X} \max_{i=1,\ldots,k}|f_i(x)|$.
\end{enumerate}

In \cite{CeDi*10} it has been shown that there is at least one distance between rank invariants, the {\em matching distance}, that has the stability property (S). An analogous stability property for a distance defined between modules is presented in \cite{Les11}.

The matching distance will be used for computations in the experiments described in Section \ref{sec:experimentation}. Until then, we will not need to specify which distance $\dm$ we are using, provided it satisfies (S).

\subsection{Persistence of simplicial complex filtrations}

We consider a simplicial complex $\cK$ consisting of closed
geometric simplices and its {\em carrier} defined by
\begin{equation} \label{eq:carrier}
K=|\cK|:=\bigcup \cK.
\end{equation}
The set of all vertices of $\cK$ is denoted by $\cV(\cK)$ or by $\cV$, if $\cK$ is clear from the context.
For $\sigma, \tau \in \cK$, the relation $\tau$ {\em is a face of} $\sigma$ is
denoted by $\tau \leq \sigma$. For proper faces, we write $\tau <
\sigma$.

In this discrete setting, we take a family $\{\cK_\alpha\}_{\alpha\in\R^k}$ of simplicial subcomplexes of $\cK$, such that $\cK_\alpha$ is a subcomplex of $\cK_\beta$, for every $\alpha\preceq \beta$. As a consequence, their carriers are nested with respect to inclusions, that is: $K_\alpha\subseteq K_{\beta}$, yielding a filtration of $K$.

In the next section we address the following problem: is any simplicial complex filtration induced by a suitable continuous function?

A positive answer to this question will allow us later to transfer the stability property of $\dm$ from the continuous to the discrete setting.

\section{From continuous to discrete vector functions}
\label{sec:simplicial}

We let $\varphi:\cV(\cK)\to \R^k$ be a vector-valued function defined on vertices. We suppose that $\varphi$ is a {\em discretization} of some continuous function $\tilde{\varphi}:K\to \R^k$. Reciprocally, $\tilde{\varphi}$ is an {\em interpolation} of $\varphi$. In this section we will simply assume that $\tilde{\varphi}$ is equal to $\varphi$ on vertices of $\cK$ but, of course, when it comes to computing, one has to set bounds for the rounding error. Although in some practical applications of persistent homology to the analysis of discrete multidimensional data $\tilde{\varphi}:K\to \R^k$ may be explicitly known, in some other cases we do not even have an explicit formula for $\tilde{\varphi}$: we assume that such a function exists, that we can estimate its modulus of uniform continuity (for the sake of simplicity, say, its Lipschitz constant), and that we can compute the values of $\tilde{\varphi}$ at grid points of arbitrary fine finite grids.

In a discrete model, we are interested in simplicial sublevel complexes
\[
\cK_{\alpha} := \{\sigma \in \cK \mid \varphi(v)\preceq \alpha
\mbox { for all vertices } v \leq \sigma\} .
\]

In Section \ref{sec:experimentation}, we compute the rank invariants for the discrete vector-valued function $\varphi$ and we use this information for computing the distance between rank invariants for their continuous interpolations. In order to do this, we need to know that there exists a continuous function which is a {\em generic interpolation} of $\varphi$, in the sense that its rank invariant is exactly equal to that of $\varphi$.

In the case $k=1$, that is, when $\varphi$ has values in $\R$, it can be shown that such an interpolation can be obtained by extending $\varphi$ to each simplex $\sigma\in\cK$ by linearity. We shall denote this interpolation by $\lphi$.
In that case, one can show that $K_{\alpha}$ is a deformation retract of $K_{\lphi\leq\alpha}$, so the inclusion of one set into another induces an isomorphism in homology.
This result belongs to ``mathematical folklore'': it is often implicitly used in computations without being proved. The arguments for that case are outlined in \cite[Section 2.5]{Mo08} and Theorem~\ref{th:sublevel-cont-vs-simplicial} we prove in this section contains this result as a special case. Unfortunately, if $k>1$, this result is no longer true as the following example shows:
\begin{ex}\label{ex:tetrahedron}  {\em
Let $K$ be the boundary of the tetrahedron shown in Figure~\ref{fig:tetrahedron}, homeomorphic to the 2D sphere.
\begin{figure}
\begin{center}
\psfrag{p}{$\lphi$}\psfrag{v0}{$v_0$}\psfrag{v1}{$v_1$}\psfrag{v2}{$v_2$}\psfrag{v3}{$v_3$}
\psfrag{p1}{$$}\psfrag{p0}{$$}\psfrag{p3}{$$}\psfrag{a}{$\alpha$}
\includegraphics[width=0.80\textwidth]{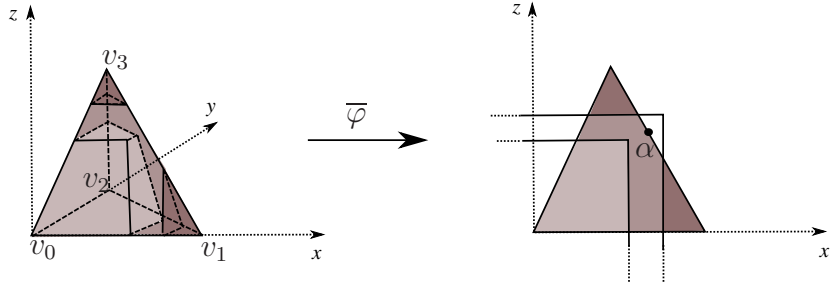}
\end{center}
\caption{The tetrahedron boundary and sketches of two sublevel sets of the linear interpolation $\lphi$ discussed in Example~\ref{ex:tetrahedron}. The values taken at the displayed edge are critical.}\label{fig:tetrahedron}
\end{figure}
The corresponding simplicial complex $\cK$ is made of all proper faces of the 3D simplex $[v_0,v_1,v_2,v_3]$ in $\R^3$, with vertices $v_0=(0,0,0)$, $v_1=(1,0,0)$, $v_2=(0,1,0)$, $v_3=(1/2,0,1)$. Hence $K=|\cK|$ is homeomorphic to a 2D sphere. Let $\varphi:K\to \R^2$ be the restriction of the linear function $\lphi$ given by $\lphi(x,y,z)=(x,z)$
to the four vertices. Let $\alpha\in \R^2$ be any value chosen so that $1/2<\alpha_1<1$ and $\alpha_2=2-2\alpha_1$.
It is easy to see that $K_\alpha=[v_0,v_2]$. Its homology is trivial.
Note that the set
$K_{\lphi\preceq \alpha}$ contains one point $x$ on the edge $[v_1,v_3]$, namely $x=(\alpha_1,0,\alpha_2)$, which closes a non-contractible path in $K_{\lphi\preceq \alpha}$. We have $H_1(K_{\lphi\preceq \alpha})\cong \F\neq 0$.

The discrepancy between the discrete and linear interpolated models seen in Example \ref{ex:tetrahedron} has been observed in applications to computer graphics and imaging, and has been recently referred to as topological aliasing.

Several interesting conclusions can be derived from this. First, $K_\alpha$ is not a deformation retract of $K_{\lphi\preceq \alpha}$. This remains true if we slightly increase the value of $\alpha$.
Secondly, if we slightly decrease $\alpha$, the set $K_\alpha$ does not change but the set $K_{\lphi\preceq \alpha}$ becomes contractible.
Hence, in the sense of Definition~\ref{def:hom-critical-value} presented further in this section, any value assumed at a point of the edge $[v_1,v_3]$ is a homological critical value. In particular, the set of such values may be uncountable. This is in contrast with the one-dimensional case, where a piecewise-linear function on a simplicial complex must have a discrete set of critical values.
 }
\end{ex}

We shall now construct a continuous function $\aphi: K\to \R^k$ called {\em axis-wise linear interpolation} of $\varphi$ which will correct the problem encountered with the linear interpolation $\lphi$ in the multidimensional case. First, given any $\sigma\in \cK$, let $\mu(\sigma) \in \R^k$ be defined by
\begin{equation} \label{eq:max-vertices}
\mu_j(\sigma) =\max\{\varphi_j(v) \mid v \mbox{ is a vertex of } \sigma \}, j=1,2,\ldots, k.
\end{equation}
Note that if $\tau\leq \sigma$, then $\mu(\tau)\preceq \mu(\sigma)$.

We will use induction on the dimension $m$ of $\sigma$ to define $\aphi: K\to \R^k$ on $\sigma$ and a point $w_\sigma\in \sigma$ with the following properties:
\begin{aenum}
\item For all $x\in \sigma$, $\aphi(x) \preceq \aphi(w_\sigma)=\mu(\sigma)$ ;
\item $\aphi$ is linear on any line segment $[w_\sigma,y]$ with $y$ on the boundary of $\sigma$.
\end{aenum}
If $m=0$, so that $\sigma=\{ v \}$ is a vertex, $\aphi(v)=\varphi(v)$ and we put $w_{\{ v \}}=v$. Let $ m>0$ and suppose $\aphi$ is constructed on simplices of lower dimensions. 
Let $\tau$ be a minimal face of $\sigma$ such that $\mu(\tau)=\mu(\sigma)$.
Consider two cases.
\begin{enumerate}
\item[(i)] 
If $\tau\ne \sigma$, then $w_\tau$ and $\aphi(w_\tau)$ are defined by the induction step. We put $w_\sigma=w_\tau$. Since $\sigma$ is convex, any $x$ in the interior of $\sigma$ is on a line segment joining $w_\sigma$ to a uniquely defined $y(x)$ on the boundary of $\sigma$. Since $\aphi(y(x))$ is defined by the induction step, we extend $\aphi$ to $[w_\sigma,y(x)]$ by linearity.

\item[(ii)] 
If $\tau=\sigma$, then let $w_\sigma$ be the barycenter of $\sigma$ and put $\aphi(w_\sigma)=\mu(\sigma)$. Again, any $x\neq w_\sigma$ in the interior of $\sigma$ is on a line segment joining $w_\sigma$ to a uniquely defined $y(x)$ on the boundary of $\sigma$ and we proceed as before.
\end{enumerate}
The property (a) follows from the fact that $\mu(\tau)\preceq \mu(\sigma)$ when $\tau\leq \sigma$, and from the linearity on joining segments. The property (b) is clear from the construction. By routine arguments from convex analysis, the point $y(x)$ on the boundary of $\sigma$ is a continuous function of $x\in \sigma \setminus \{w_\sigma$\}, and the constructed function is continuous on $\sigma$. Since we proceeded by induction on the dimension of $\sigma$, the definitions on any two simplices coincide on their common face, so $\aphi$ extends continuously to $K$. The property (b) implies that if $k=1$, and in certain cases of vector valued functions, $\aphi$ is equal to  $\lphi$, namely:

\begin{aenum}
\item[(c)] $\aphi$ is piecewise linear on each simplex $\sigma$. In addition, if $w_\tau$ is a vertex of $\tau$ for each $\tau\leq \sigma$, then it is linear on $\sigma$.
\end{aenum}

The difference between the piecewise linear and the axis-wise linear interpolations $\lphi$ and $\aphi$ is illustrated in Figure \ref{fig:axiswise} for a 1-simplex $\sigma=[v_0,v_1]$ and a function $\varphi$ defined on vertices.

\begin{figure}
\begin{center}
\psfrag{v0}{$v_0$}\psfrag{v1}{$v_1$}\psfrag{v2}{$v_2$}\psfrag{w}{$w_\sigma$}\psfrag{lpw}{$\lphi(w_\sigma)$}\psfrag{apw}{$\aphi(w_\sigma)$}\psfrag{p}{$\varphi$}\psfrag{a}{$\varphi_1(v_1)$}\psfrag{b}{$\varphi_1(v_0)$}\psfrag{c}{$\varphi_2(v_0)$}\psfrag{d}{$\varphi_2(v_1)$}\psfrag{p1}{$\varphi_1$}\psfrag{p2}{$\varphi_2$}
\includegraphics[width=0.60\textwidth]{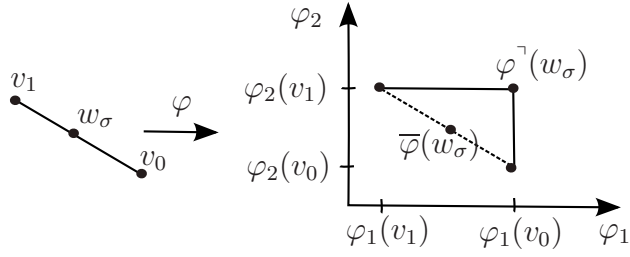}
\end{center}
\caption{The linear (dashed line) and axis-wise linear (continuous line) interpolations of a function $\varphi$ defined on the vertices of the simplex $\sigma=[v_0,v_1]$ with values in $\R^2$.   }\label{fig:axiswise}
\end{figure}

\begin{lem} \label{lem:corner}
The following statements hold:
\begin{enumerate}
\item[(i)] For any $\alpha\in\R^k$, $K_\alpha \subset K_{\aphi \preceq \alpha}$.
\item[(ii)] Let $\sigma \in \cK$ and $\alpha\in\R^k$. If $\sigma \cap K_{\aphi \preceq \alpha}\neq \emptyset$, then $\sigma$ has at least one vertex in $K_\alpha$.
\end{enumerate}
\end{lem}

\begin{proof}
{\em (i)} Let $\sigma \in \cK_\alpha$. It is clear from (\ref{eq:max-vertices}) that $\mu(\sigma)\preceq \alpha$. It follows from the property (a) in the definition of $\aphi$ that $\sigma\subset K_{\aphi \preceq \alpha}$.\\~\\
{\em (ii)} We follow the induction steps in the construction of $\aphi$. If $\dim(\sigma)=0$, $\sigma$ is a vertex and there is nothing to prove. Let $\dim(\sigma)=m>0$ and suppose the statement is proved for lower dimensions. Let $x\in \sigma \cap K_{\aphi \preceq \alpha}$. If $x=w_\sigma$, then $w_\sigma\in K_{\aphi \preceq \alpha}$. By the property (a) of $\aphi$, all $\sigma$ is in $K_\alpha$. If $x\neq w_\sigma$, then $x$ is on a line segment joining the point $w_\sigma$ of $\sigma$ with a point $y(x)$ of an $(m-1)$--simplex $\tau < \sigma$, where $\aphi(y(x))$ is defined by the induction hypothesis. We know that $\aphi$ is extended linearly to the line segment $[w_\sigma,y(x)]$. Also, $\aphi(y(x)) \preceq \aphi(w_\sigma)$ by the property (a). Hence $\aphi(y(x))\preceq \aphi(x) \preceq \alpha$. It remains to use the induction hypothesis for $y(x)$ and $\tau$ to deduce that $\tau$ has a vertex in $K_\alpha$.
\end{proof}

\begin{thm} \label{th:sublevel-cont-vs-simplicial}
For any $\alpha \in \R^k$, $K_\alpha$ is a strong deformation
retract of $K_{\aphi \preceq \alpha}$. Consequently, the inclusion
$K_\alpha \hookrightarrow K_{\aphi \preceq \alpha}$ induces an
isomorphism in homology.
\end{thm}

\begin{proof}
Note that $K_{\aphi\preceq \alpha}$ is contained in a union of simplices $\sigma \in \cK$ such that
\begin{equation} \label{eq:sigma-sublevel}
\sigma_{\aphi \preceq \alpha}:=\sigma \cap K_{\aphi \preceq \alpha}\neq \emptyset.
\end{equation}
Given any such $\sigma$, consider the simplex $\sigma_\alpha$ defined as the convex hull of the set of vertices $v$ of $\sigma$ such that $\varphi(v)\preceq \alpha$. By the hypothesis on $\sigma$ and by Lemma~\ref{lem:corner}{\em (ii)}, $\sigma_\alpha \neq \emptyset$. 
Given any $\sigma\in\cK$ for which $\sigma_{\aphi \preceq \alpha}\neq\emptyset$,
we shall define a strong deformation retraction
\[
H_\sigma:\sigma_{\aphi \preceq \alpha} \times [0,1]\to
\sigma_{\aphi \preceq \alpha}
\]
with $r=H(\cdot,1)$ being a retraction of $\sigma_{\aphi\preceq \alpha}$ onto $\sigma_\alpha$.

The construction goes by induction on the dimension $m$ of $\sigma$ following the induction steps in the construction of the function $\aphi$. If $\dim(\sigma)=0$, $\sigma$ is a vertex and there is nothing to prove. 
Now let $m>0$. Suppose that the deformation retraction $H_\tau: \tau_{\aphi \preceq \alpha} \times [0,1]\to
\tau_{\aphi \preceq \alpha}$ is defined for simplices $\tau$ of dimension $m'<m$ with $\tau_{\aphi \preceq \alpha}\neq \emptyset$ in such a way that $H_\tau(x,t)=x$ for any $(x,t)\in \tau_{\aphi \preceq \alpha} \times [0,1/2^{m'}]$, and the values of $\aphi$ on $H_\tau(x,t)$ are decreasing with $t$. By ``decreasing'' we mean the weak inequality ``$\preceq$''. This hypothesis guarantees that the deformation has values in the set $\tau_{\aphi \preceq \alpha}$.

Let $x\in \sigma_{\aphi \preceq \alpha}$. If $x$ is on a boundary of $\sigma$, we define $H_\sigma(x,t)=H_\tau(x,t)$, where $\tau$ is the smallest face of $\sigma$ containing $x$ and $H_\tau$ is defined by the induction hypothesis. Suppose $x$ is in the interior of $\sigma$. Let $w_\sigma$ and $y(x)$ be the points identified in the definition of $\aphi$. Note that, if $\aphi(w_\sigma)\preceq \alpha$, then $\sigma_{\aphi \preceq \alpha}=\sigma$, hence the deformation must be defined as the identity map for each $t$:
\[
H_\sigma(x,t):=x \mbox{ for all } (x,t)\in \sigma \times [0,1].
\]
Therefore, we may suppose that $w_\sigma\notin \sigma_{\aphi \preceq \alpha}$. Consider the smallest face $\tau$ of $\sigma $ containing $y(x)$. Since $y(x)$ is on the boundary of $\sigma$, $\tau$ is a proper face of $\sigma$ of dimension, say, $m'<m$. By the construction of $\aphi$,
\begin{equation} \label{eq:inequalities}
\aphi(y(x))\preceq \aphi(x) \preceq \aphi(w_\sigma).
\end{equation}
Since $\aphi(x)\preceq \alpha$, we get $\aphi(y(x))\preceq \alpha$ so $y(x)\in \tau_{\aphi \preceq \alpha}\neq \emptyset$. By the induction hypothesis, a deformation retraction
\[
H_\tau:\tau_{\aphi \preceq \alpha} \times [0,1]\to
\tau_{\aphi \preceq \alpha} \subset \sigma_{\aphi \preceq \alpha}
\]
is defined so that the values of $\aphi$ on $H_\tau(x,t)$ decrease with $t$, and $H_\tau(x,t)=x$ for $t\in [0,1/2^{m'}]$. 

For any $t\in [0,1]$ and for $x$ in the interior of $\sigma$ we define 
\[
H_{\sigma}(x,t) := \left\{
\begin{array}{ll}
x  &  \text{ if }\; 0 \leq t < 1/2^m\\
(2^m t-1)y(x)-(2^m t-2)x   &  \text{ if }\; 1/2^m \leq t < 1/2^{m-1}\\
H_\tau(y(x),t) & \text{ if }\; 1/2^{m-1} \leq t \le 1
\end{array}
\right. .
\]

It is easily checked that $H_{{\sigma}}(x,1/2^m)=x$, $H_{{\sigma}}(x,1/2^{m-1})=y(x)$. Since $\aphi$ is linear on $[w_\sigma,y(x)]$, the inequality (\ref{eq:inequalities}) implies that the values of $\aphi$ on $H_{{\sigma}}(x,t)$ decrease with $t$. 

Thus we have defined $H_\sigma$ both when $x$ is on the boundary of $\sigma$ and when it is in the interior of $\sigma$. By construction, for every $x\in\sigma_{\aphi \preceq \alpha}$, $H_\sigma(x,0)=x$, and $H_{{\sigma}}(x,1)$ belongs to $\sigma_\alpha$, and moreover, for every $x\in\sigma_\alpha$, $H_{{\sigma}}(x,1)=x$. In order to conclude that $H_\sigma$ is a deformation retraction of $\sigma_{\aphi \preceq \alpha}$ onto $\sigma_\alpha$ we must prove that $H_{\sigma}$ is continuous. The continuity at a given point $(x_0,t_0)$ with $x_0$ in the interior of $\sigma$ follows from the continuity of $y(x)$ in $x$. The continuity at $(x_0,t_0)$ with $x_0$ on the boundary of $\sigma$ follows from the condition  that $H_\tau(x,t)=x$ for any $t\in [0,1/2^{m'}]$ and from the induction hypothesis. 

In order to continuously extend $H_\sigma$ to a deformation
\[
H:K_{\aphi \preceq \alpha} \times [0,1]\to K_{\aphi \preceq \alpha},
\]
it is enough to prove that, given two simplices $\sigma_1$ and $\sigma_2$ intersecting $K_{\aphi \preceq \alpha}$ and $\tau=\sigma_1 \cap \sigma_2$, the maps $H_{\sigma_1}$ and $H_{\sigma_2}$ agree at any $x\in \tau_{\aphi \preceq \alpha}$. It is clear from the definition that $H_{\sigma_1}(x,t)=H_{\sigma_2}(x,t)=H_{\tau}(x,t)$ for $x\in \tau$ and for all $t$, provided that $H_\tau$ is defined. But this is true, because $x\in\tau_{\aphi \preceq \alpha}$, so this is a nonempty set.
\end{proof}

In the next section, we use Theorem~\ref{th:sublevel-cont-vs-simplicial} to show that any distance between rank invariants of continuous functions that has property (S) can be approximated by the distance between rank invariants of discrete functions. We end this section with another application of Theorem~\ref{th:sublevel-cont-vs-simplicial} of interest in itself: a theorem on the structure of the set of critical values of the axis-wise interpolation $\aphi$. The following definition generalizes the notion of homological critical value given in \cite{CoEdHa07} to vector functions. In plain words we call homological critical any value $\alpha$ for which any sufficiently small neighborhood contains two values whose sublevel sets are included one into the other but cannot be retracted one onto the other. Neighborhoods are taken with respect to the norm $\|\alpha\|= \max_{j=1,2,\ldots, k}|\alpha_j|$ in $\R^k$.

\begin{defn} \label{def:hom-critical-value} {\em
Let $\mphi:K\to \R^k$ be a continuous vector function. A value $\alpha\in \R^k$ is a {\em homological critical value} of $\mphi$ if there exists an integer $q$ such that, for all sufficiently small real values $\epsilon>0$, two values $\alpha',\alpha''\in \R^k$ can be found with $\alpha'\preceq \alpha\preceq \alpha''$, $\|\alpha'-\alpha\|<\epsilon$, $\|\alpha''-\alpha\|<\epsilon$, such that the map
\[
H_q(K_{\mphi \preceq \alpha'})\rightarrow H_q(K_{\mphi \preceq \alpha''})
\]
induced by the inclusion $K_{\mphi \preceq \alpha'} \hookrightarrow K_{\mphi \preceq \alpha''}$ is not an isomorphism. If this condition fails, $\alpha$ is called a {\em homological regular value}.
 }
\end{defn}

Also note that, by the long exact sequence for the relative homology (see
e.g.\ \cite[Chapter 9]{KacMisMro04}) the critical value definition is equivalent to the
following condition on the graded relative homology
\[
H_*(K_{\mphi\preceq\alpha' },K_{\mphi\preceq\alpha''})\neq
0.
\]
For any $j=1,2,\ldots, k$ and a vertex $v\in \cV(\cK)$, consider the hyperplane of $\R^k$ given by the equation $\alpha_j=\aphi_j(v)$ and a positive closed cone $C_j(v)$ contained in it, given by the formula
\[
C_j(v):=\{\alpha\in\R^k\mid \alpha_j=\aphi_j(v) \mbox{ and } \alpha_i\geq \aphi_i(v) \mbox{ for all }i=1,2,\ldots, k\}.
\]
\begin{thm} \label{th:hom-crit-C}
The set of homological critical values of $\aphi$ is contained in the finite union of the described cones, namely, in the set
\[
C:=\bigcup\{C_j(v)\mid v\in\cV(\cK) \mbox{ and } j=1,2,\ldots, k\}.
\]
\end{thm}

\begin{proof}
Consider any $\alpha\notin C$. We need to show that $\alpha$ is a homological regular value. Since $C$ is a closed set, an $\epsilon>0$ exists such that the set $\bar Q(\alpha,\epsilon)=\{\beta\in \R^k\mid \|\alpha-\beta\|\le \epsilon\}$ does not meet $C$.
If $\|\alpha-\beta\|\le\epsilon$, then
\begin{equation} \label{eq:alpha-pm-eps}
K_{\beta}=K_{\alpha}.
\end{equation}
Indeed, if this were not true, the segment joining $\alpha$ and $\beta$ should contain a point of $C$, against the choice of $\epsilon$.

Now, let us assume that $\alpha'\preceq\varphi(v)\preceq\alpha''$,
$\|\alpha-\alpha'\|\le\epsilon$ and $\|\alpha-\alpha''\|\le\epsilon$.
It follows from equation (\ref{eq:alpha-pm-eps}) and from Theorem~\ref{th:sublevel-cont-vs-simplicial} that the inclusions
$i':K_{\alpha}=K_{\alpha'} \hookrightarrow K_{\aphi \preceq \alpha'}$
and
$i'':K_{\alpha}=K_{\alpha''} \hookrightarrow K_{\aphi \preceq \alpha''}$
induce isomorphisms in homology. The inclusion
$i^{(\alpha',\alpha'')}:K_{\aphi \preceq \alpha'} \hookrightarrow K_{\aphi \preceq \alpha''}$
can be written as
$i^{(\alpha',\alpha'')}=i'' \circ r'$, where $r'$ is the retraction homotopically inverse to $i'$.
By the functoriality of homology,
$H_*(i^{(\alpha',\alpha'')})=H_*(i'') \circ H_*(r')$,
hence it is also an isomorphism.
\end{proof}

For the sake of visualization, in Figure \ref{fig:C} the set $C$ is shown in a simple case.

\begin{figure}
\begin{center}
\psfrag{p1}{$\varphi_1$}\psfrag{p2}{$\varphi_2$} \psfrag{p1a}{$\varphi_1(v_0)$}\psfrag{p1b}{$\varphi_1(v_1)$}\psfrag{p1c}{$\varphi_1(v_2)$}\psfrag{p1d}{$\varphi_1(v_3)$}
\psfrag{p2a}{$\varphi_2(v_0)$}\psfrag{p2b}{$\varphi_2(v_1)$}\psfrag{p2c}{$\varphi_2(v_2)$}\psfrag{p2d}{$\varphi_2(v_3)$}
\includegraphics[width=0.60\textwidth]{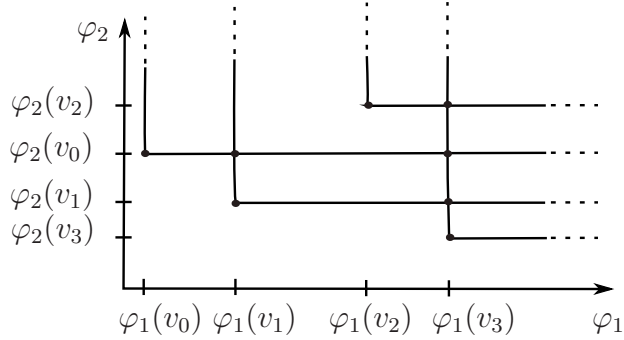}
\end{center}
\caption{The set $C$ defined in Theorem \ref{th:hom-crit-C} is the union of closed cones with vertices at the values taken by $\varphi$. The set $\Lambda$ introduced in Proposition \ref{prop:Lambda} is a finite set whose elements are the highlighted points.}
\end{figure} \label{fig:C}

From the formula for $C$, we instantly get an analogy of a well-known result from differential geometry \cite{Sma73}.

\begin{cor}\label{cor:nowhere-dense}
The set of homological critical values of $\aphi$ is a nowhere dense set in $\R^k$. Moreover its $k$-dimensional Lebesgue measure is zero.
\end{cor}

\section{Approximation of distances between rank invariants}
\label{Sec:approximation}

The goal of this section is to show that shape comparison by persistent homology of vector functions is numerically stable. In this passage from real (continuous) objects to their discretizations, the approximation error does not grow to be much larger when we compute the distance $\dm$ between the rank invariants of $\varphi$ and $\psi$ instead of the distance $\dm$ between the rank invariants of $f$ and $g$ (Theorem \ref{thm:main}). To this end, the stability property (S) of $\dm$ defined in Section~\ref{sec:stability-property}
in the continuous setting is crucial.

A description of this approximation procedure in concrete examples together with experiments exploiting the numerical stability of the comparison by persistent homology will be given in Section \ref{sec:experimentation}.

We end the section by showing that the set of homological critical values, although uncountable, admits a finite representative set.\\

We start from the following approximation lemma. It may happen that $\dm(\rho_\mphi,\rho_\mpsi)$ and $\dm(\rho_\aphi,\rho_\apsi)$ are equal to $+\infty$. In the case of the matching distance, this occurs when $H_*(K)\ne H_*(L)$. In such a case, we adopt the convention $\infty-\infty =0$.

\begin{lem}\label{lem:approx}
Let $\mphi:K\to \R^k$, $\mpsi:L\to \R^k$ be two continuous measuring functions on the carrier of complexes $\cK$ and $\cL$.
For any $\epsilon>0$, there exists $\delta>0$ such that if
\begin{equation} \label{eq:simplex-delta}
\max \{\diam \sigma\mid \sigma\in \cK \mbox{ or } \sigma\in \cL\} < \delta
\end{equation}
then
\begin{equation} \label{eq:dmatch-epsilon}
|\dm(\rho_\mphi,\rho_\mpsi)-\dm(\rho_\aphi,\rho_\apsi) |< \epsilon.
\end{equation}
\end{lem}

\begin{proof}
Since $K$ and $L$ are compact, $\mphi$, $\mpsi$, $\aphi$, and $\apsi$ are uniformly continuous. Hence for any $\epsilon>0$, there exists $\delta>0$ such that if (\ref{eq:simplex-delta}) is satisfied then
\begin{equation} \label{eq:simplex-phi}
\max \{\diam \mphi(\sigma)\mid \sigma\in \cK \} < \epsilon/4
\end{equation}
and the same inequality holds for $\mpsi$, $\aphi$, and $\apsi$. The diameters of $\mphi(\sigma)$ and $\sigma$ are measured with respect to the maximum norm in the respective ambient spaces. Since $\varphi$ is the restriction of $\mphi$ to the vertices, and $\aphi$ interpolates $\varphi$ on the vertices, given any $x\in\sigma\in\cK$, and any vertex $v$ of $\sigma$, from (\ref{eq:simplex-phi}) we get 
\begin{equation} \label{eq:measuring-corner-eps}
\|\mphi(x) -\aphi(x)\|\leq \|\mphi(x) -\varphi(v)\|+\|\varphi(v)-\aphi(x)\|<\epsilon/2.
\end{equation}
Hence, by the choice of the maximum norm in $\R^k$, $\|\mphi -\aphi\|_\infty <\epsilon/2$. By the same arguments, $\|\mpsi -\apsi\|_\infty <\epsilon/2$. By the stability property (S) of $\dm$,
\begin{eqnarray*}
\dm(\rho_\mphi,\rho_\mpsi) & \leq & \dm(\rho_\mphi,\rho_\aphi)+\dm(\rho_\aphi,\rho_\apsi)+\dm(\rho_\apsi,\rho_\mpsi)\\
 & \leq & \|\mphi - \aphi\|_\infty +\dm(\rho_\aphi,\rho_\apsi) + \|\apsi - \mpsi\|_\infty \\
 & < & \dm(\rho_\aphi,\rho_\apsi) + \epsilon.
\end{eqnarray*}
Reversing the roles of $\mphi$, $\mpsi$ and $\aphi$, $\apsi$, we get $\dm(\rho_\aphi,\rho_\apsi)<\dm(\rho_\mphi,\rho_\mpsi)+\epsilon$ and the conclusion follows.
\end{proof}


Knowing Lemma~\ref{lem:approx}, we now turn our attention to computing $\dm(\rho_\aphi,\rho_\apsi)$.

\medskip

 The following definition sets the notation for the rank invariant of the simplicial complex filtration obtained from a discrete map $\varphi$. Next, we show that this definition gives a rank invariant coinciding with the rank invariant of the continuous function $\aphi$. Thus it is a first step in the passage from the stability of rank invariants for continuous functions to that of discrete ones. Moreover, this definition is the one which we use to implement the reduction algorithm of \cite{CaDiFe10} in our computations in Section \ref{sec:experimentation}.

\begin{defn} \label{def:rank-inv-cont-alpha}
 {\em
Consider the discrete map $\varphi:\cV(\cK)\to \R^k$ defined on vertices of a simplicial complex $\cK$. The {\em $q$'th real space variable rank invariant} or, shortly, {\em $q$'th real rank invariant} of $\varphi$ is the function $\rho_\varphi^q:\Delta^k_+ \to \N$ defined on each pair $(\alpha,\beta)\in \Delta^k_+$ as the rank of the map
\[
H_q(j^{(\alpha,\beta)}): H_q(K_\alpha) \to
H_q(K_\beta)
\]
induced by the inclusion map $j^{(\alpha,\beta)}:K_\alpha \hookrightarrow K_\beta$ on simplicial sublevel sets.
 }
\end{defn}

\begin{thm} \label{th:cont-simpl-rank}
Given any discrete function $\varphi:\cV(\cK)\to \R^k$ on the set of vertices of a simplicial complex $\cK$ and its axis-wise interpolation $\aphi$, we have the equality of $q$'th real rank invariants
\[
\rho_\varphi^q = \rho_\aphi^q.
\]
\end{thm}

\begin{proof} Consider any $(\alpha,\beta)\in \Delta^k_+$, the inclusion maps $i^{(\alpha,\beta)}:K_{\aphi\leq
\alpha}\hookrightarrow K_{\aphi\leq \beta}$, and $j^{(\alpha,\beta)}:K_\alpha \hookrightarrow K_\beta$.
Theorem~\ref{th:sublevel-cont-vs-simplicial} implies that for every $q\in\Z$ we have the following commutative diagram
\[
\begin{array}{ccc}
H_q(K_{\aphi\preceq \alpha}) & \mapright{H_q(i^{(\alpha,\beta)})} & H_q(K_{\aphi\preceq \beta}) \\
\mapdown{\cong} & & \mapdown{\cong} \\
H_q(K_\alpha) & \mapright{H_q(j^{(\alpha,\beta)})} & H_q(K_\beta)
\end{array}
\]
where the vertical arrows are the isomorphisms induced by the corresponding retractions. Thus $\rank H_q(i^{(\alpha,\beta)}) = \rank H_q(j^{(\alpha,\beta)})$.

\end{proof}

In the sequel, we will once again use $\rho_\varphi$ to refer to real rank invariants of arbitrary order. In conclusion we obtain that the distance between the rank invariants of two measuring (or interpolation) functions can be approximated using only the corresponding simplicial sublevel sets.

\begin{cor} \label{cor:dmatch-approx}
Let $\mphi:K\to \R^k$, $\mpsi:L\to \R^k$ be two continuous measuring functions on the carriers of complexes $\cK$ and $\cL$ and let $\varphi:\cV(\cK)\to \R^k$, $\psi:\cV(\cL)\to \R^k$ be the discretizations of $\mphi$ and $\mpsi$ on the sets of vertices of $\cK$ and $\cL$, respectively. For any $\epsilon>0$ there exists $\delta>0$ such that if
\[
\max \{\diam \sigma\mid \sigma\in \cK \mbox{ or } \sigma\in \cL\} < \delta
\]
then
\begin{equation} \label{eq:dmatch-approx}
|\dm(\rho_\mphi,\rho_\mpsi)-\dm(\rho_\varphi,\rho_\psi) |< \epsilon.
\end{equation}
\end{cor}

\begin{proof}
Immediate from Lemma \ref{lem:approx} and Theorem \ref{th:cont-simpl-rank}.
\end{proof}

We are now ready to give the main result of this section.

\begin{thm}\label{thm:main}
Let $X$ and $Y$ be homeomorphic triangulable topological spaces, and let $f:X\to \R^k$, $g:Y\to \R^k$ be continuous functions. Let $(K,\mphi)$ and $(L,\mpsi)$, with $K$ and $L$ carriers of complexes $\cK'$ and $\cL'$, and $\mphi:K\to\R^k$, $\mpsi:L\to \R^k$ continuous measuring functions, approximate $(X,f)$ and $(Y,g)$, respectively, in the following sense:  For a fixed $\epsilon>0$, there exist a homeomorphism $\xi:K\to X$ with $\|\mphi-f\circ \xi\|_{\infty}\le\epsilon/4$ and a homeomorphism $\zeta:L\to Y$ with $\|\mpsi-g\circ \zeta\|_{\infty}\le\epsilon/4$. Then, for any sufficiently fine subdivision $\cK$ of $\cK'$ and $\cL$ of $\cL'$,
$$\left|\dm(\rho_f,\rho_g)-\dm(\rho_\varphi,\rho_\psi)\right|\le\epsilon,$$
$\varphi:\cV(\cK)\to \R^k$, $\psi:\cV(\cL)\to \R^k$ being restrictions of $\mphi$ and $\mpsi$ on the set of vertices of $\cK$ and $\cL$, respectively.
\end{thm}

\begin{proof}
By the triangle inequality
$$\dm(\rho_f,\rho_g)\le \dm(\rho_f,\rho_{f\circ \xi})+\dm(\rho_{f\circ \xi},\rho_\mphi)+\dm(\rho_\mphi,\rho_\mpsi)+\dm(\rho_\mpsi,\rho_{g\circ \zeta})+\dm(\rho_{g\circ \zeta},\rho_g).$$
Since $\rho_f=\rho_{f\circ \xi}$ and $\rho_g=\rho_{g\circ \zeta}$, we have $\dm(\rho_f,\rho_{f\circ \xi})=0$ and $\dm(\rho_{g\circ \zeta},\rho_g)=0$. Moreover, by the stability property (S), since $\|\mphi-f\circ \xi\|_{\infty}\le\epsilon/4$ and $\|\mpsi-g\circ \zeta\|_{\infty}\le\epsilon/4$, we have $\dm(\rho_{f\circ \xi},\rho_\mphi)\le\epsilon/4$ and $\dm(\rho_\mpsi,\rho_{g\circ \zeta})\le\epsilon/4$. Therefore,
$$\dm(\rho_f,\rho_g)\le \dm(\rho_\mphi,\rho_\mpsi) +\epsilon/2.$$
By Corollary \ref{cor:dmatch-approx}, there exists $\delta>0$ such that, if $\cK$ and $\cL$ are subdivisions of $\cK'$ and $\cL'$ with $\max \{\diam \sigma\mid \sigma\in \cK \mbox{ or } \sigma\in \cL\} < \delta$, then $\dm(\rho_\mphi,\rho_\mpsi)\le \dm(\rho_\varphi,\rho_\psi) + \epsilon/2$. In conclusion we have proved that $\dm(\rho_f,\rho_g)\le \dm(\rho_\varphi,\rho_\psi) +\epsilon.$

Reversing the roles of $f$, $g$ and $\varphi$, $\psi$, we get
$\dm(\rho_\varphi,\rho_\psi)\le \dm(\rho_f,\rho_g) +\epsilon$, yielding the claim.
\end{proof}

We turn now to the question of the structure of the critical set of $\aphi$. Recall from the previous section that when $k>1$, the set of homological critical values of a function on $K$ with values in $\R^k$ may be an uncountable set, although contained in a nowhere dense set $C$ by Theorem~\ref{th:hom-crit-C}. However, the family $\{K_\alpha\}_{\alpha\in \R^k}$ of all subcomplexes of $\cK$ is finite. Thus there exists a finite representative set $\Lambda\subset \R^k$ for $C$ as the following proposition states.

\begin{prop} \label{prop:Lambda}
For any $\alpha\in C$, there exists $\lambda$ in
\[
\Lambda=\{\lambda\in C \mid \forall\; j=1,2,\ldots, k,\; \exists\; v\in\cV(\cK)\, : \lambda_j=\varphi_j(v)\}
\]
such that $K_\alpha=K_\lambda$.
\end{prop}

\begin{proof}
Since $\alpha \in C$, $\cV(\cK_\alpha)\ne \emptyset$ and there exists $j$ such that $\alpha_j=\varphi_j(\bar v_j)$ for some $\bar v_j\in \cV(\cK)$, and $\alpha_i\ge \varphi_i(\bar v_j)$, for $1\le i\le k$. For each $i\ne j$, let us take a vertex $\bar v_i\in \cV(\cK_\alpha)$ such that $\varphi_i(\bar v_i)\ge \varphi_i(v)$ for every $v\in \cV(\alpha)$. Now we set $\lambda=(\lambda_1,\ldots, \lambda_k)$, with $\lambda_j=\varphi_j(\bar v_j)$. By construction $\lambda$ belongs to $\Lambda$. Furthermore, it holds that $K_\alpha=K_\lambda$. Indeed, obviously, $K_\lambda\subseteq K_\alpha$. Moreover, for every $v\in K_\alpha$, by definition of $\bar v_i$ it holds that $\varphi_i(\bar v_i)\ge \varphi_i(v)$ for $1\le i\le k$. Equivalently, $\lambda_i\ge \varphi_i(v)$ for $1\le i\le k$, implying that $v\in K_\lambda$.
\end{proof}

The structure of set $\Lambda$ is visualized in Figure \ref{fig:C}.
The previous proposition prompts for the following definition.
\begin{defn} \label{def:disrete-rank-inv}
 {\em
Consider the discrete map $\varphi:\cV(\cK)\to \R^k$ defined on vertices of $\cK$. The {\em discrete rank invariant} of $\varphi$ is the restriction of the real rank invariant $\rho_\varphi$ to the finite domain $\Lambda^2_+:=\Delta^k_+\cap (\Lambda\times\Lambda)$.
 }
\end{defn}

Definition~\ref{def:disrete-rank-inv} gives a discrete rank invariant which is similar to the one defined in \cite{CaZo07,CaSiZo09}, except for the fact that we are using a different homological structure.

Defining a distance $\dm$ directly on the basis of $k$-dimensional rank invariants would be a task impossible to accomplish. Even when $k=2$, a pair of complexes $\cK$ and $\cL$ with an order of thousand vertices would result in computing ranks of millions of maps induced by inclusions. This motivates the one-dimensional reduction method described in the next section to compute the matching distance.

\section{Algorithm and experimentation}
\label{sec:experimentation}

For experimentation purposes, we now fix the distance between rank invariants that we will use to be the matching distance $\dmm$ defined in \cite{CaDiFe10}.

The one-dimensional reduction method presented in \cite{CaDiFe10} to compute the matching distance consists of applying the one-dimensional rank invariant along the lines $t \mapsto \vec{b} + t \vec{l}$ parameterized by $t$ and determined by pairs of vectors $(\vec{l},\vec{b})$ in a chosen grid in $\R^k\times \R^k$, where $\vec{b}$ is an initial point and $\vec{l}$ directs the line. It is assumed that all components of $\vec{l}$ are positive, and that $\vec{l} \cdot {\bf 1}=1$, $\vec{b} \cdot {\bf 1}=0$, where ${\bf 1}=(1,1,\ldots, 1)$. For all $(\alpha, \beta) \in \Delta^k_+$, there exists a unique such pair, which will be called {\em linearly admissible pair} or simply {\em admissible pair}, the set of which will be denoted $Ladm_k$. Also denote by $g_{(\vec{l},\vec{b})}:K\to \R$ and $h_{(\vec{l},\vec{b})}:L\to \R$ the one-dimensional functions given by $g_{(\vec{l},\vec{b})}(x)=\max_i\, (\varphi_i^\urcorner(x)-b_i)/l_i$ and $h_{(\vec{l},\vec{b})}(x)=\max_i\,(\psi_i^\urcorner(x)-b_i)/l_i$, where $l_i$ and $b_i$ are the $i$-th components of $\vec{l}$ and $\vec{b}$, respectively. For ease of notation, the pair $(\vec{l},\vec{b})$ may be left out of $g$ and $h$ if it is unambiguous. By \cite[Lemma 1]{CaDiFe10}, if $\alpha=\vec{b} + s \vec{l}$, then
\[
K_{\aphi\preceq \alpha}=K_{g\leq s}.
\]
This and Theorem~\ref{th:cont-simpl-rank} implies
\begin{cor} \label{th:lb-line-simpl-rank}
Consider $(\alpha,\beta)=(\vec{b} + s \vec{l},\vec{b} + t \vec{l})\in \Delta^k_+$, for some $(s,t)\in\Delta^1_+$. Then
\[
\rho_\varphi(\alpha,\beta) = \rho_g(s,t).
\]
\end{cor}
The above theorem shows that it is legitimate to apply the reduction method of \cite{CaDiFe10} to simplicial sublevel sets. Following \cite[Definition 11]{CaDiFe10}, we define the {\em multidimensional matching distance} between the rank invariants $\rho_\varphi$ and $\rho_\psi$ to be
$$\dmm(\rho_\varphi,\rho_\psi) = \sup_{(\vec{l},\vec{b})\in Ladm_k} \min_{i=1,\ldots, k} l_i \,
\dmm(\rho_{g_{(\vec{l},\vec{b})}},\rho_{h_{(\vec{l},\vec{b})}}).$$
In this section, the value $\min_{i=1,\ldots, k} l_i \,
\dmm(\rho_{g_{(\vec{l},\vec{b})}},\rho_{g_{(\vec{l},\vec{b})}})$ will be denoted
$\mathrm{d_m} (\rho_\varphi,\rho_\psi)$ or $\mathrm{d_m}_{(\vec{l},\vec{b})} (\rho_\varphi,\rho_\psi)$ and referred to as the \textit{rescaled one-dimensional matching distance}. The computational problem is, given a threshold value $\epsilon > 0$, computing
an approximate matching distance $\widetilde{\dmm}(\rho_\varphi,\rho_\psi)$ on a suitable finite subset
$A \subset Ladm_k$ such that
\begin{equation} \label{eq:errorbound}
\widetilde{\dmm}(\rho_\varphi,\rho_\psi) \leq \dmm(\rho_\varphi,\rho_\psi) \leq\widetilde{\dmm}(\rho_\varphi,\rho_\psi) + \epsilon.
\end{equation}

\subsection{Algorithm} \label{subsect:algorithm}

Our algorithm's inputs consist of lists of simplices of $\cK$ and $\cL$ of highest dimension together with their adjacency relations and vertices, and of the values of normalized measuring functions $\varphi: \cV(\cK) \to \R^2$ and $\psi: \cV(\cL) \to \R^2$, as well as a tolerance $\epsilon$.\footnote{Due to the finite precision of computer arithmetic, the codomain of the functions $\varphi_i$ and $\psi_i$ is in reality $10^{-p} \,\Z$ rather than $\R$. In our computations we tended to use $p=6$, that is, a precision of up to six digits after the decimal point.} In our computations, we have confined ourselves to the case where $\cK$ and $\cL$ are triangular meshes. Its output is an approximate matching distance $\widetilde{\dmm}(\rho_\varphi,\rho_\psi)$. To compute the one-dimensional persistent homology on admissible pairs, we use the persistent homology
software JPlex \cite{JPlex}. By default, JPlex computes the persistent Betti numbers over $\Z_{11}$ of a
discretely indexed filtration of simplicial complexes. We build this filtration by adding simplices
in the following recursive way. We first order the values attained by the one-dimensional measuring
function $g$ in increasing order, $\{g_1, \ldots, g_N\}$. A finite filtration $\{\cK_1, \ldots,
\cK_N\}$ is then built by inserting simplices $\sigma\in\cK$. If $\sigma = \{v\}$, where $v$ is a vertex,
we put $\{v\}$ into $\cK_i$ if $g(v) \leq g_i$. Otherwise, $\sigma\in\cK_i$ if all its vertices are
in $\cK_i$. Similarly, the function $h$ is used to build a finite filtration $\{\cL_1, \ldots,
\cL_M\}$ using simplices of $\cL$.

The set of admissible pairs $Ladm_2$ is the set of quadruples $(a,1-a,b,-b)\in\R^2\times\R^2$ such that $0<a<1$. As described in
\cite[Remark 3.2]{BiCe*10}, it is possible to avoid computation of
the one-dimensional matching distance over a large portion of $Ladm_2$. Since the functions $\varphi$ and $\psi$ are normalized, $C = \max\{\|\varphi\|,
\|\psi\|\} = 1$. Let $Ladm_2^*$ be the set of admissible pairs such that $|b|<1$. Then, to compute the maximal
value of $\dmm(\rho_g,\rho_h)$ over $Ladm_2 \backslash Ladm_2^*$, it is sufficient to consider the
two admissible pairs $(a,1-a,b,-b) = (1/2,1/2,2,-2)$ and $(1/2,1/2,-2,2)$. The details
can be found in \cite{BiCe*10}.

It follows from a generalization of the Error Bound Theorem (\cite[Theorem 3.4]{BiCe*10} and \cite{CeFr}) to persistent homology of arbitrary order that if 
for $(\vec{l},\vec{b})$ and $(\vec{l}',\vec{b}') \in Ladm_2$, $\|(\vec{l},\vec{b}) - (\vec{l}',\vec{b}')\| \leq \delta$, then for normalized functions $\varphi$ and $\psi$
$$|\,\mathrm{d_m}_{(\vec{l},\vec{b})} (\rho_\varphi,\rho_\psi) - \mathrm{d_m}_{(\vec{l}',\vec{b}')}
(\rho_\varphi,\rho_\psi)| \leq 18 \delta.$$
This suggests that in order to satisfy Equation~(\ref{eq:errorbound}), it suffices to choose admissible pairs $(\vec{l},\vec{b})\in Ladm_2^*$ at a distance within $\epsilon/9$ of each other, guaranteeing that every member of $Ladm_2^*$ is within $\epsilon/18$ of a tested pair. In practice, our algorithm is reminiscent of the grid algorithm shown in Section 3 of \cite{BiCe*10}, in the sense that we take pairs at a distance of $1/2^N$ of each other with $N$ sufficiently large. We observe that the set $Ladm_2$ is in bijective correspondence with $\{(a,b) \in\R^2 \, | \, a\in (0,1), b\in\R\}$, and so we will speak of computing $\mathrm{d_m} (\rho_\varphi,\rho_\psi)$ at a point $P = (a,b)$ of the preceding set. The lattice of points on which we compute this rescaled matching distance is chosen as follows: choose $N \in \N$ such that $1/2^N \leq \epsilon/18$, and choose $P_{ij} = (a_i, b_j), i=0, \ldots, 2^N - 1, j=0, \ldots, 2^{N+1} - 1$ such that $a_i = (2i+1)/2^{N+1}, b_j = 1 - (2j+1)/2^{N+1}$.

\subsection{Examples of topological aliasing}\label{subsect:aliasing}

Our experimentations have been made on triangular meshes of compact 2D surfaces. In doing so, the influence on experimental results of the concept of topological aliasing discussed in Section \ref{sec:simplicial} became apparent. Namely, we used our algorithm to compare in a pairwise manner 10 cat models, a selection of which is found in Figure~\ref{fig:models}. We used for $\varphi_i$ and $\psi_i$, $i=1,2$, the following functions. Assume that the model $\cK$ is such that its vertex set $\cV(\cK) = \{v_1, \ldots, v_n\}$ and compute the following principal vector:
$$\vec{w} = \frac{\sum_{i=1}^n (v_i-c)\|v_i-c\|_2}{\sum_{i=1}^n\|v_i-c\|_2^2},$$
where $c$ is the centre of mass of $K$ defined by taking the weighted average of the centres of each triangle. Let $d$ be the line passing through $c$ having $\vec{w}$ as its direction vector, and let $\pi$ be the plane passing through $c$ having $\vec{w}$ as its normal vector.
 We defined
$$\varphi_1(v_i) = 1 - \frac{\mathrm{dist}(v_i,d)}{\max_{j=1, \ldots, n}\mathrm{dist}(v_j,d)}$$
and
$$\varphi_2(v_i) = 1 - \frac{\mathrm{dist}(v_i,\pi)}{\max_{j=1, \ldots, n}\mathrm{dist}(v_j,\pi)},$$
where $\mathrm{dist}(v,d)$ and $\mathrm{dist}(v,\pi)$ are defined in the usual way, as the minimal Euclidean distance between $v$ and the points on $d$ or $\pi$. The functions $\psi_1$ and $\psi_2$ were defined similarly using the model $\cL$. We then repeated the same procedure on the barycentric subdivisions of the models, with the value of the function at the new vertices defined using the linear interpolant. We found out that in this case the computed matching distance did not always yield the same result as when using the original unsubdivided models. However, replacing the linear interpolant by the axis-wise linear interpolant allowed us to retrieve the same results.

\begin{figure}[ht]
\begin{center}
\begin{tabular}{ccc}
\includegraphics[width=0.35\textwidth]{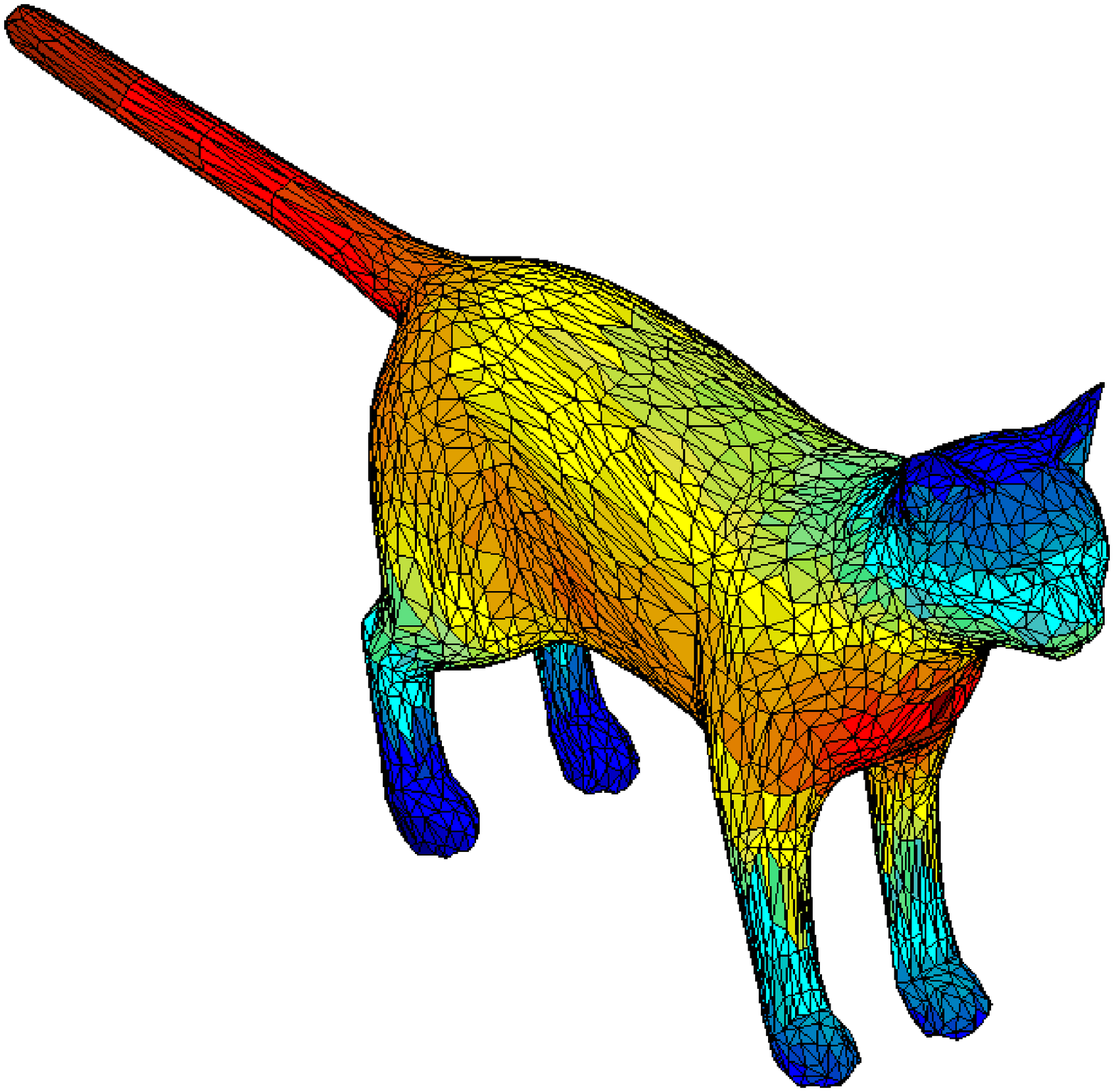} &
\includegraphics[width=0.35\textwidth]{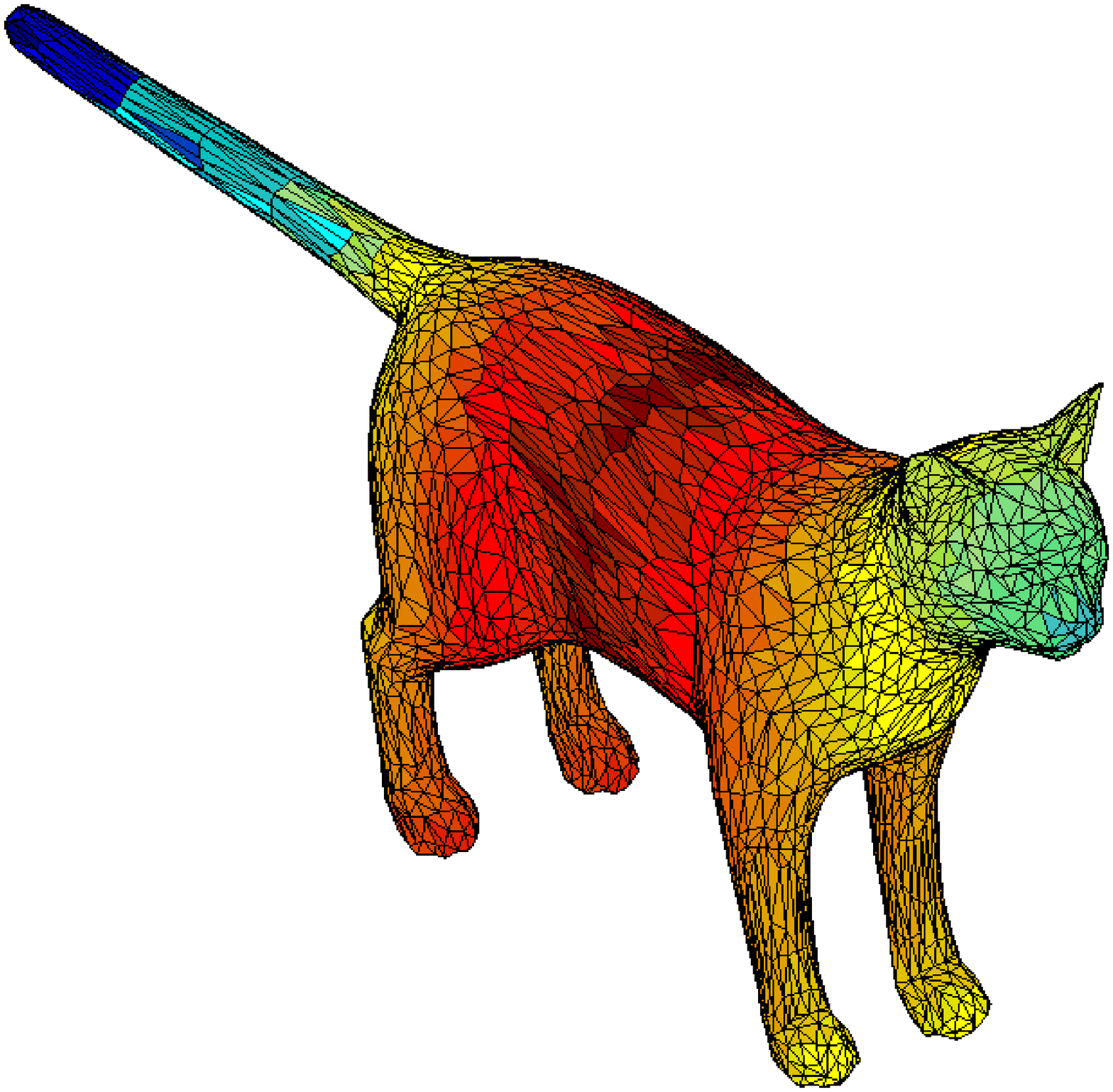} & \multirow{5}{*}{\includegraphics[width=0.15\textwidth]{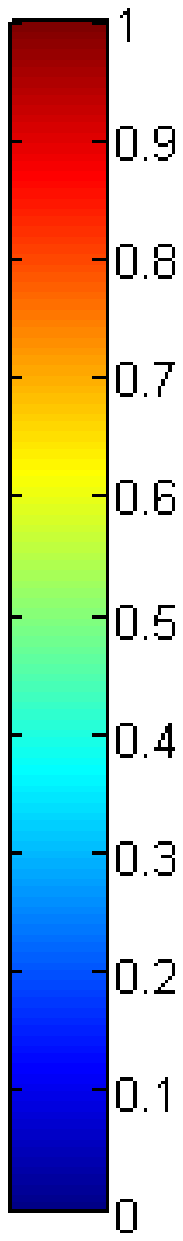}}\\
\includegraphics[width=0.35\textwidth]{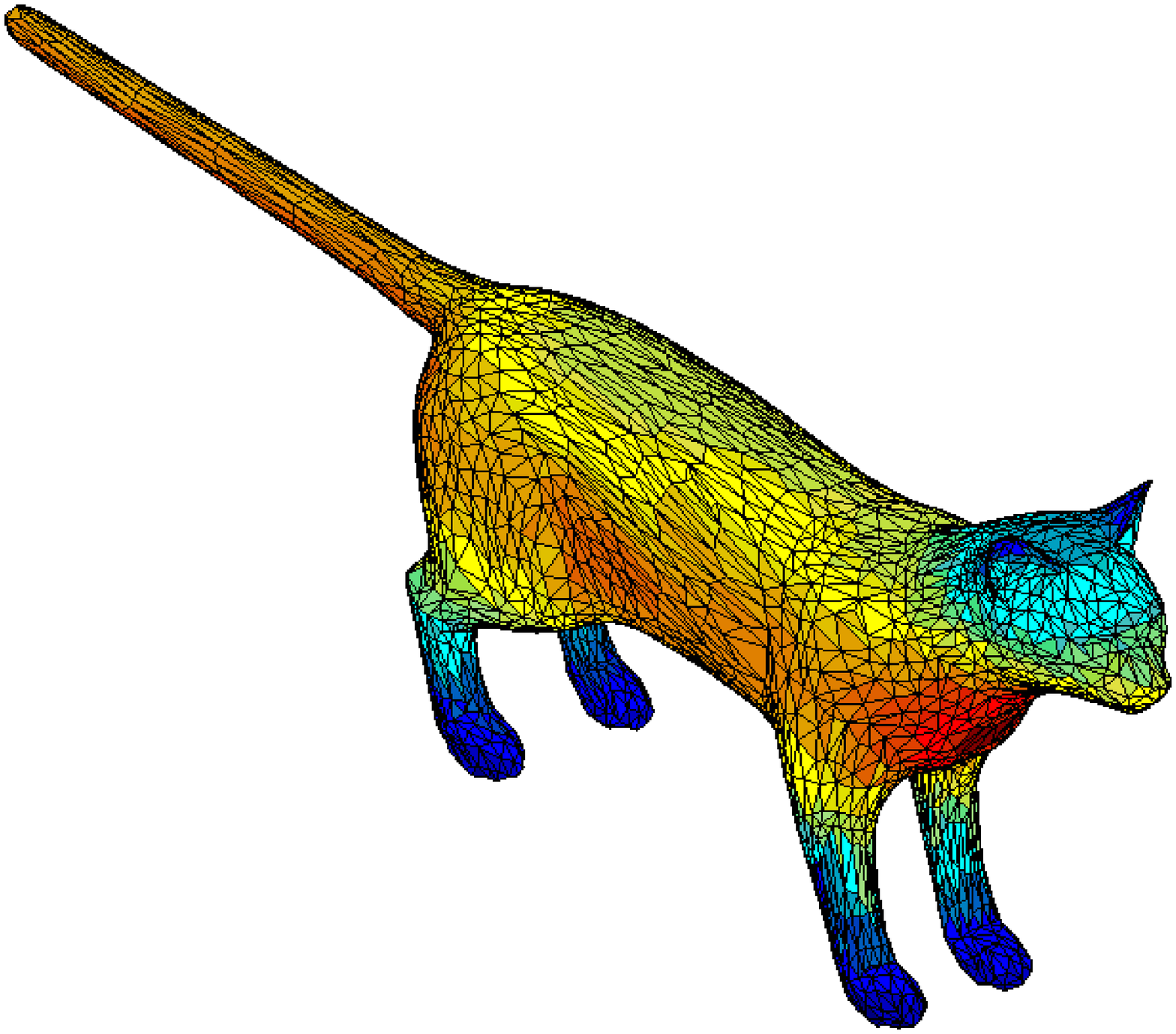}&
\includegraphics[width=0.35\textwidth]{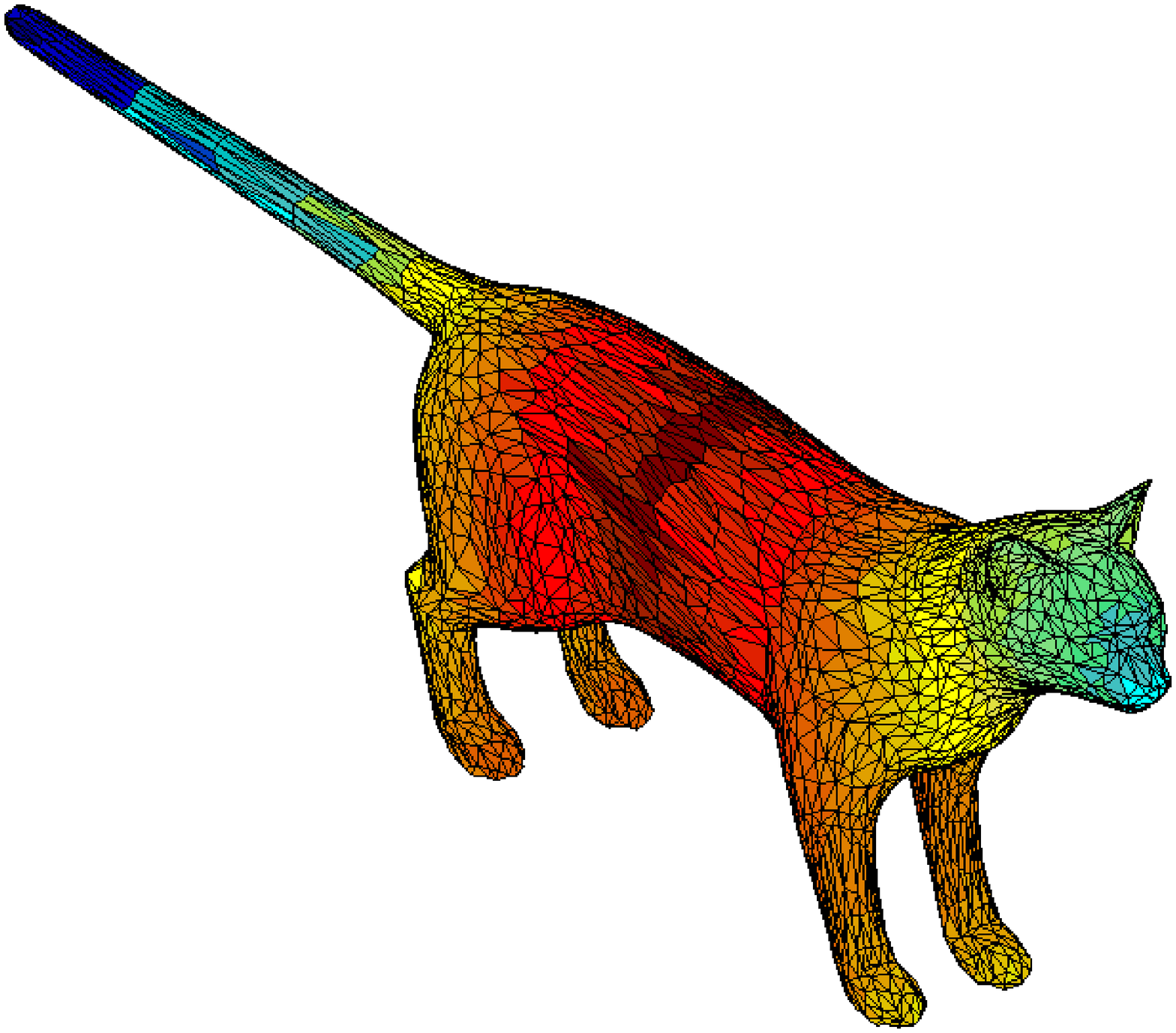}&\\
\includegraphics[width=0.35\textwidth]{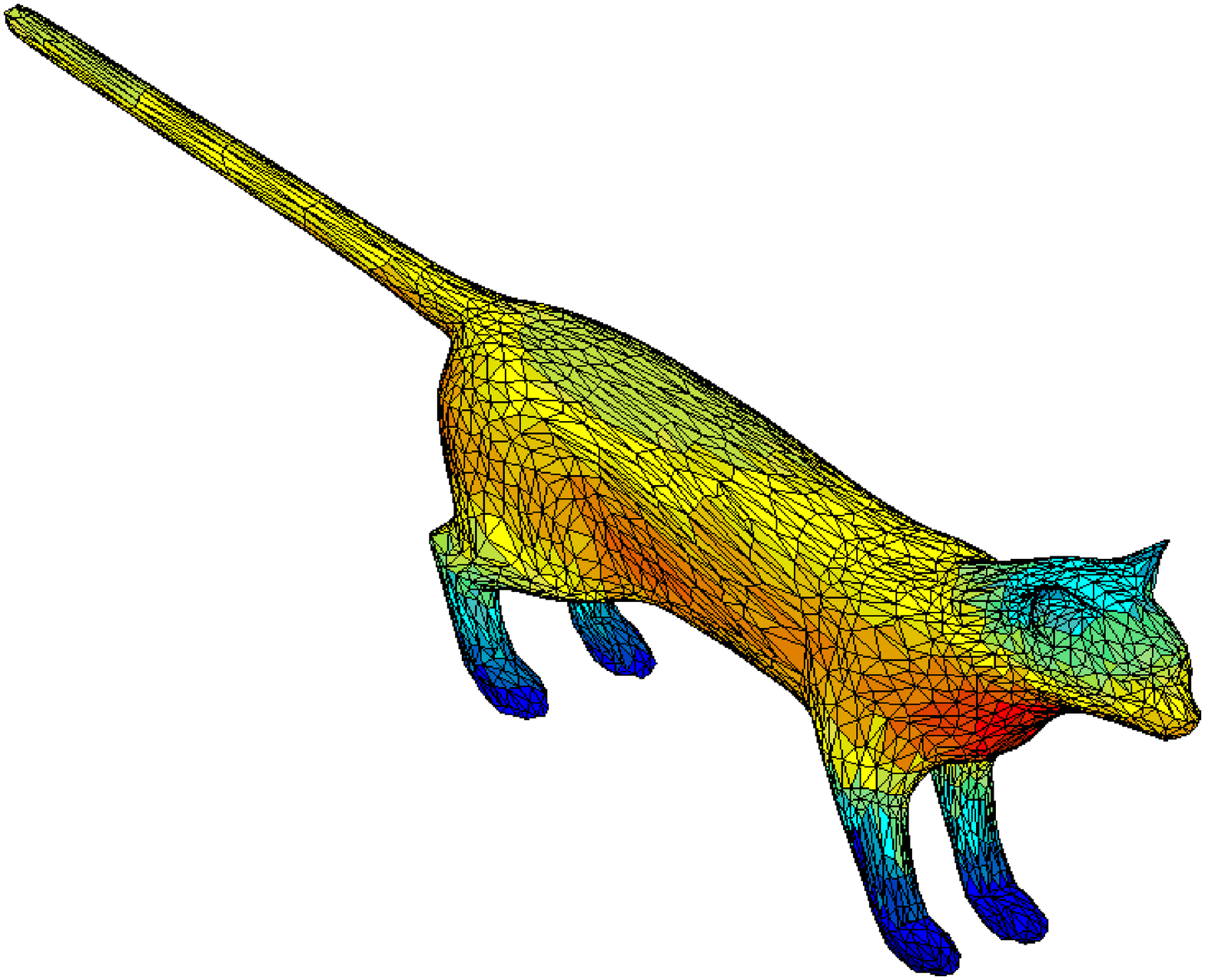}&
\includegraphics[width=0.35\textwidth]{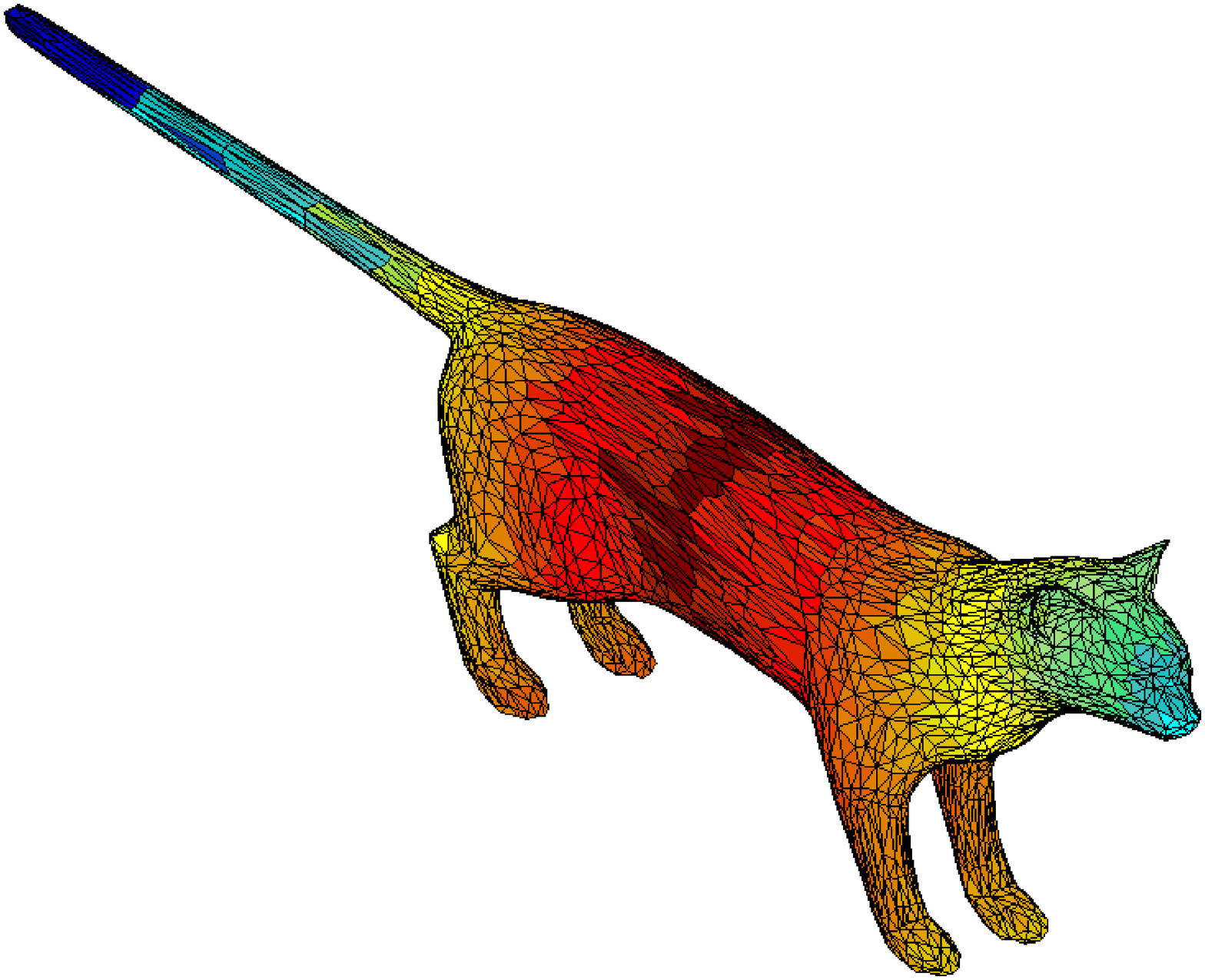}&\\
\includegraphics[width=0.35\textwidth]{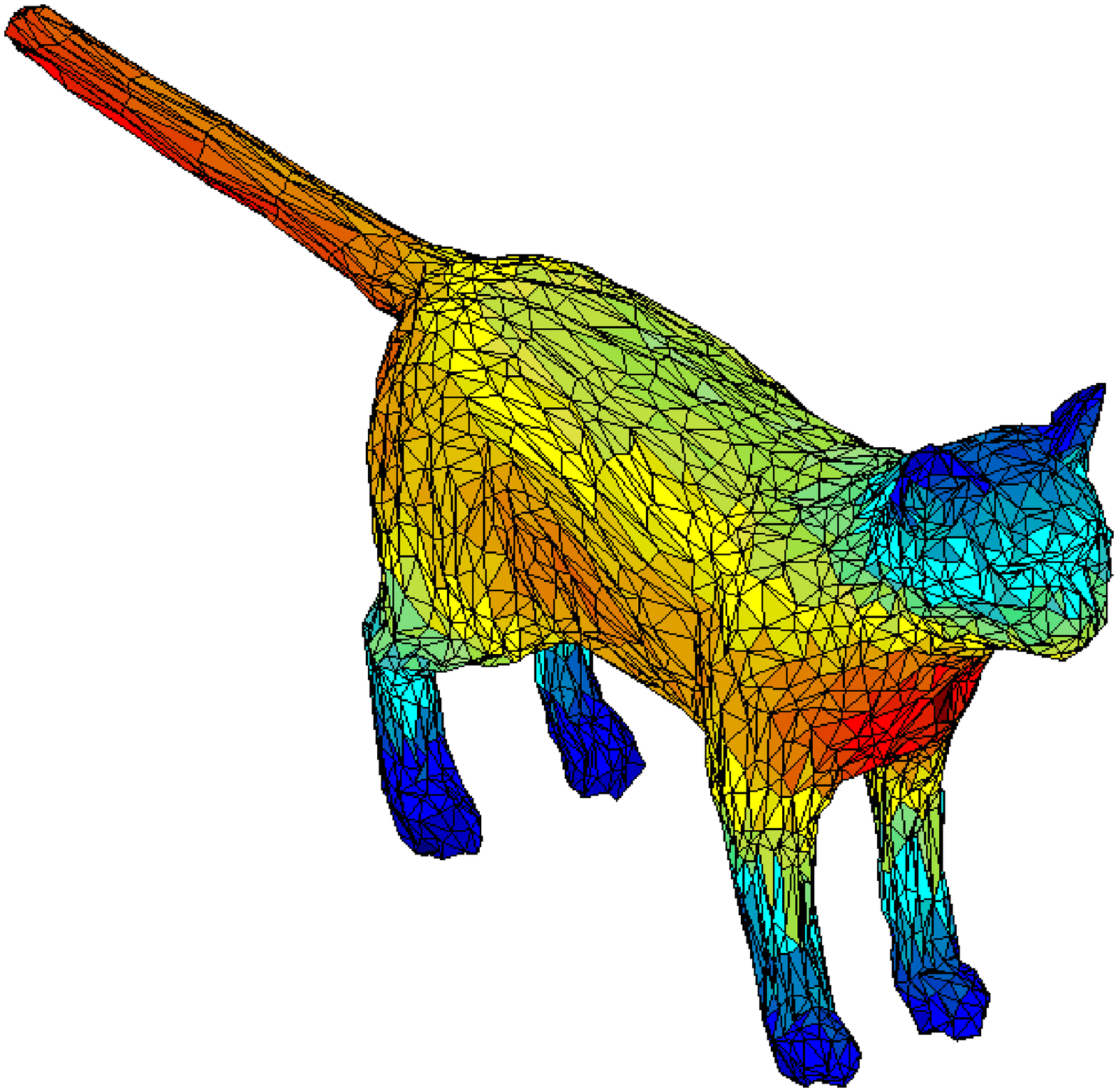}&
\includegraphics[width=0.35\textwidth]{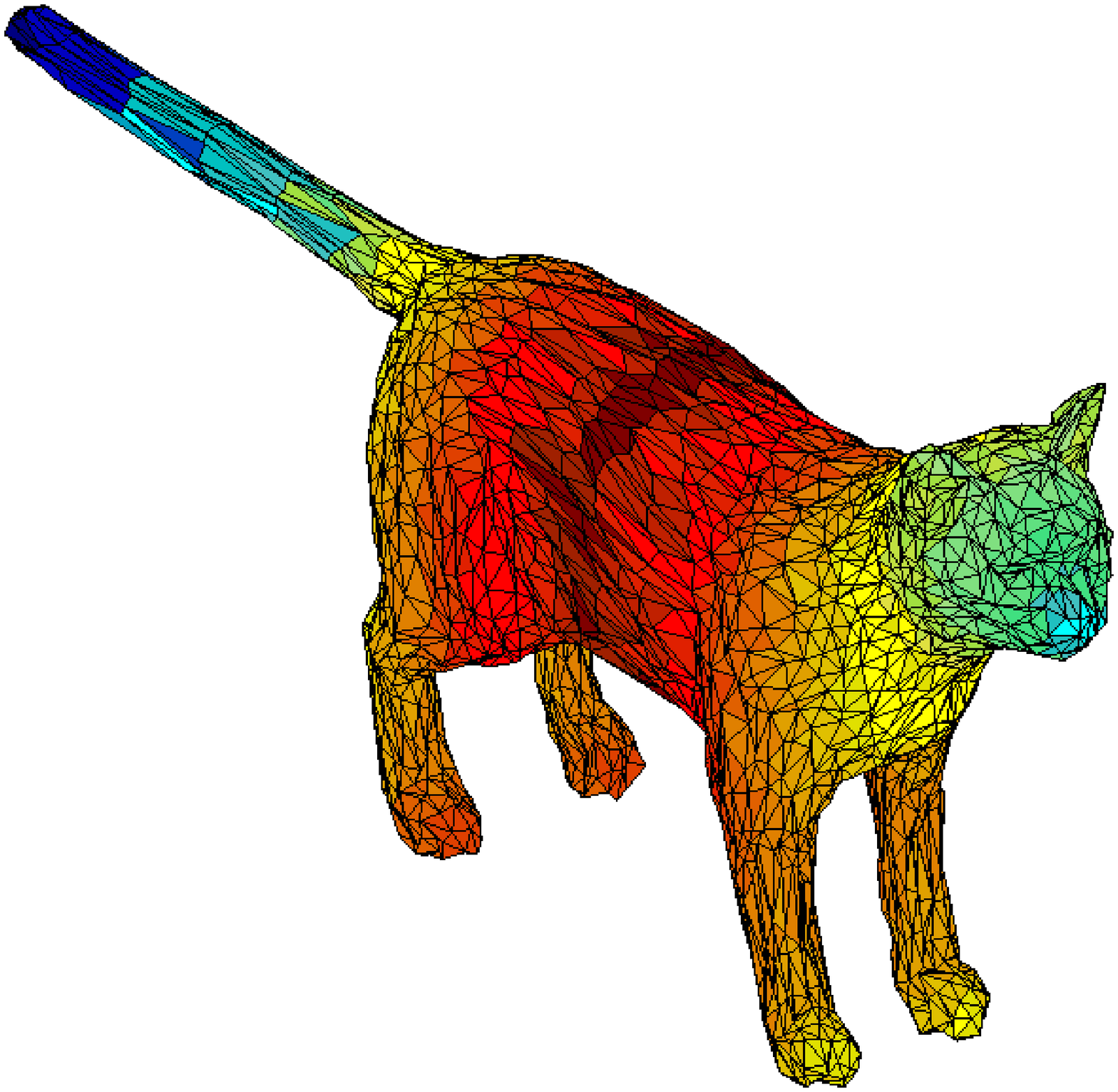}&\\
\includegraphics[width=0.35\textwidth]{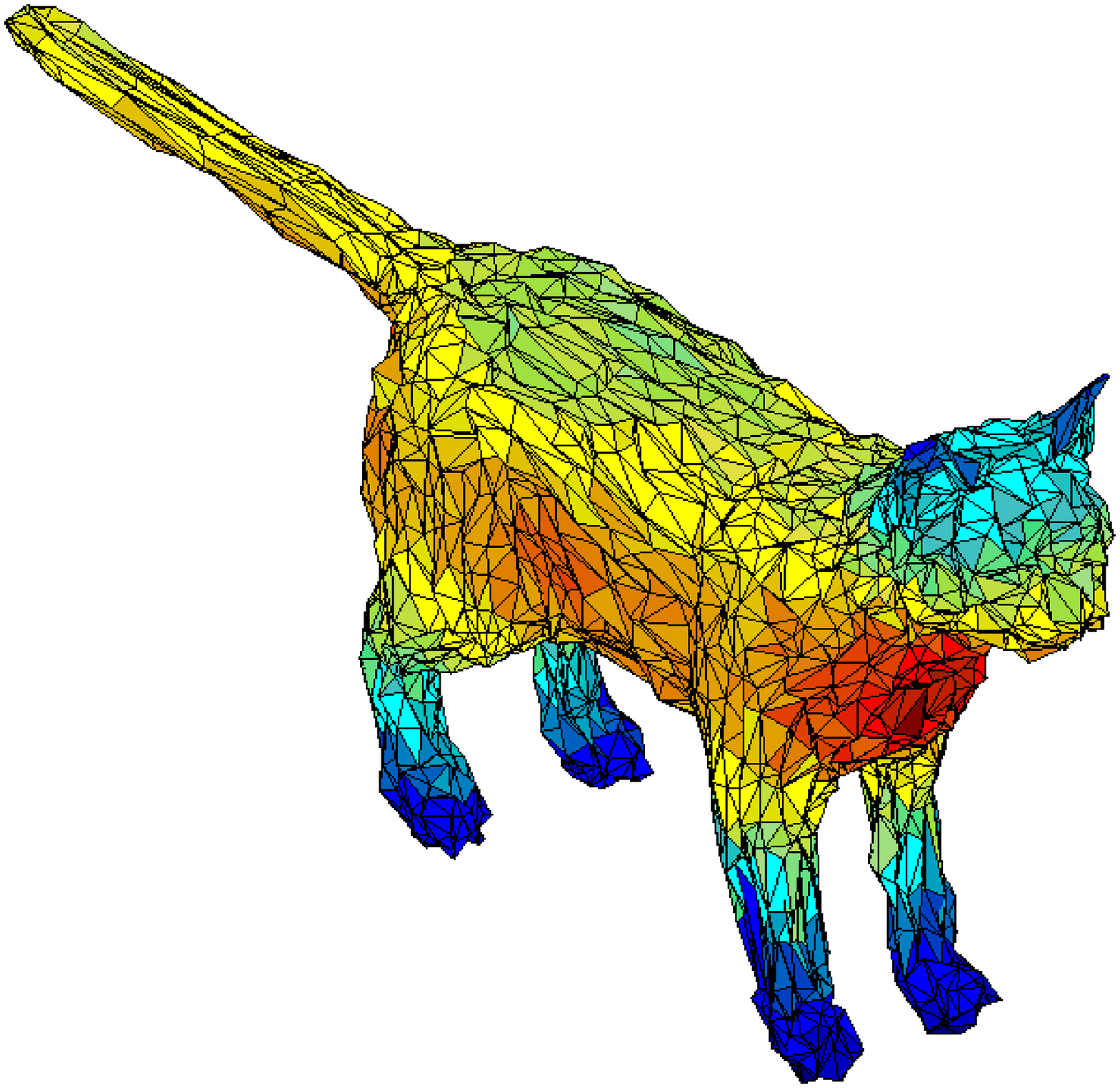}&
\includegraphics[width=0.35\textwidth]{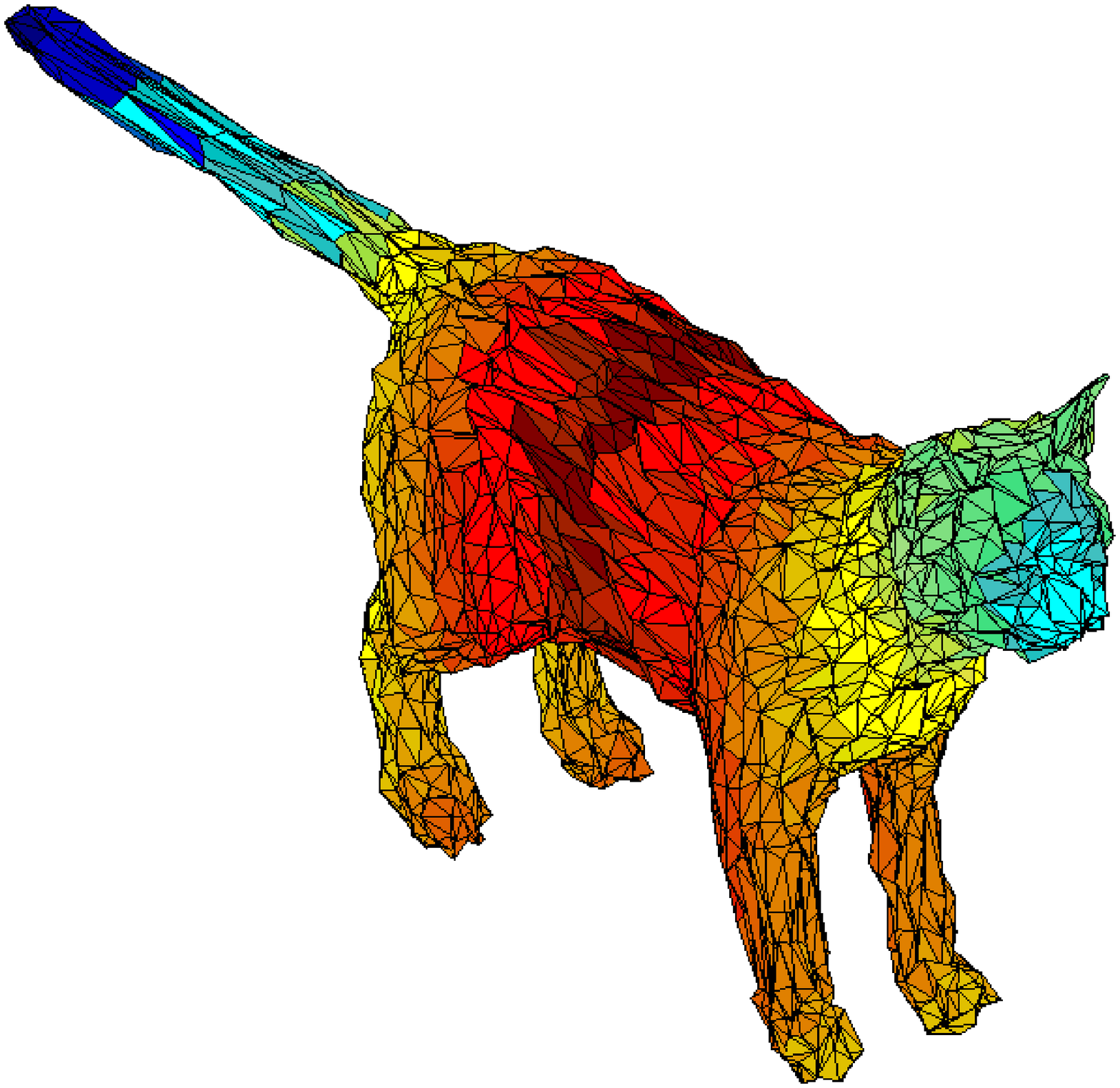}&
\end{tabular}
\end{center}
\caption{The models cat0, cat0-tran1-1, cat0-tran1-2, cat0-tran2-1 and cat0-tran2-2 are shown along with the values of $\varphi_1$ and $\varphi_2$. Models are courtesy of the authors of \cite{BiCe*10}.} \label{fig:models}
\end{figure}

We show in Table~\ref{table:results} a selected subset of our results, and in Figure~\ref{fig:models} images of the five models for which results are shown. The first two rows of numbers in each table represent the 1D matching distance computed using each of $\varphi_i$ and $\psi_i$ for $i=1, 2$ respectively, while the last three rows represent the approximated 2D matching distance computed for three different tolerance levels $\epsilon$. The column named ``Nonsub'' shows the distances computed on the original model, while ``Linear'' and ``Axis-wise'' show those computed on the subdivided models with respectively the linear and axis-wise linear interpolants. ``Diff'' and ``\% Diff'' show the difference and relative difference between the matching distance results for the unsubdivided models and subdivided models with linear interpolation.

\begin{table}[ht]
\begin{center}
\begin{tabular}{|c|c|c|c|c|c|}\cline{2-6}
\multicolumn{1}{c}{ } & \multicolumn{5}{|c|}{cat0 vs. cat0-tran1-1}\\
\cline{2-6}
\multicolumn{1}{c|}{ } & Nonsub & Linear & Axis-wise & Diff & \% Diff\\
\hline
\hline
& 0.031129 & 0.031129 & 0.031129 & \textit{0.000000} & \textit{0.000000}\\
\cline{2-6}
& 0.039497 & 0.039497 & 0.039497 & \textit{0.000000} & \textit{0.000000}\\
\cline{2-6}
$H_1$ & 0.039497 & 0.039497 & 0.039497 & \textit{0.000000} & \textit{0.000000}\\
\cline{2-6}
& 0.046150 & 0.039497 & 0.046150 & \textit{-0.006653} & \textit{-16.844317}\\
\cline{2-6}
& \textbf{0.046150} & \textbf{0.040576} & \textbf{0.046150} & \textbf{\textit{-0.005574}} & \textbf{\textit{-13.737185}}\\
\hline
\hline
& 0.118165 & 0.118165 & 0.118165 & \textit{0.000000} & \textit{0.000000}\\
\cline{2-6}
& 0.032043 & 0.032043 & 0.032043 & \textit{0.000000} & \textit{0.000000}\\
\cline{2-6}
$H_0$ & 0.194217 & 0.177001 & 0.194217 & \textit{-0.017216} & \textit{-9.726499}\\
\cline{2-6}
& 0.224227 & 0.203102 & 0.224227 & \textit{-0.021125} & \textit{-10.401178}\\
\cline{2-6}
& \textbf{0.225394} & \textbf{0.207266} & \textbf{0.225394} & \textbf{\textit{-0.018128}} & \textbf{\textit{-8.746249}}\\
\hline
\end{tabular}
\\[0.5cm]
\begin{tabular}{|c|c|c|c|c|c|}\cline{2-6}
\multicolumn{1}{c}{ } & \multicolumn{5}{|c|}{cat0-tran1-2 vs. cat0-tran2-1}\\
\cline{2-6}
\multicolumn{1}{c|}{ } & Nonsub & Linear & Axis-wise & Diff & \% Diff\\
\hline
\hline
& 0.017272 & 0.017272 & 0.017272 & \textit{0.000000} & \textit{0.000000}\\
\cline{2-6}
& 0.026101 & 0.026101 & 0.026101 & \textit{0.000000} & \textit{0.000000}\\
\cline{2-6}
$H_1$ & 0.026101 & 0.028686 & 0.026101 & \textit{0.002585} & \textit{9.903835}\\
\cline{2-6}
& 0.034314 & 0.028686 & 0.034314 & \textit{-0.005628} & \textit{-19.619327}\\
\cline{2-6}
& \textbf{0.034314} & \textbf{0.029188} & \textbf{0.034314} & \textbf{\textit{-0.005126}} & \textbf{\textit{-17.562012}}\\
\hline
\hline
& 0.182985 & 0.182985 & 0.182985 & \textit{0.000000} & \textit{0.000000}\\
\cline{2-6}
& 0.018951 & 0.018951 & 0.018951 & \textit{0.000000} & \textit{0.000000}\\
\cline{2-6}
$H_0$ & 0.192872 & 0.188365 & 0.192872 & \textit{-0.004507} & \textit{-2.392695}\\
\cline{2-6}
& 0.207480 & 0.202844 & 0.207480 & \textit{-0.004636} & \textit{-2.285500}\\
\cline{2-6}
& \textbf{0.208451} & \textbf{0.204511} & \textbf{0.208451} & \textbf{\textit{-0.003940}} & \textbf{\textit{-1.926547}}\\
\hline
\end{tabular}
\\[0.5cm]
\begin{tabular}{|c|c|c|c|c|c|}\cline{2-6}
\multicolumn{1}{c}{ } & \multicolumn{5}{|c|}{cat0-tran2-1 vs. cat0-tran2-2}\\
\cline{2-6}
\multicolumn{1}{c|}{ } & Nonsub & Linear & Axis-wise & Diff & \% Diff\\
\hline
\hline
& 0.022001 & 0.022001 & 0.022001 & \textit{0.000000} & \textit{0.000000}\\
\cline{2-6}
& 0.034288 & 0.034288 & 0.034288 & \textit{0.000000} & \textit{0.000000}\\
\cline{2-6}
$H_1$ & 0.034288 & 0.034288 & 0.034288 & \textit{0.000000} & \textit{0.000000}\\
\cline{2-6}
& 0.045545 & 0.035702 & 0.045545 & \textit{-0.009843} & \textit{-27.569884}\\
\cline{2-6}
& \textbf{0.045545} & \textbf{0.037061} & \textbf{0.045545} & \textbf{\textit{-0.008484}} & \textbf{\textit{-22.891989}}\\
\hline
\hline
& 0.095677 & 0.095677 & 0.095677 & \textit{0.000000} & \textit{0.000000}\\
\cline{2-6}
& 0.032966 & 0.032966 & 0.032966 & \textit{0.000000} & \textit{0.000000}\\
\cline{2-6}
$H_0$ & 0.178776 & 0.182322 & 0.178776 & \textit{0.003546} & \textit{1.983488}\\
\cline{2-6}
& 0.202770 & 0.196977 & 0.202770 & \textit{-0.005793} & \textit{-2.940952}\\
\cline{2-6}
& \textbf{0.212733} & \textbf{0.208097} & \textbf{0.212733} & \textbf{\textit{-0.004636}} & \textbf{\textit{-2.227807}}\\
\hline
\end{tabular}
\caption{The approximated matching distance computed by our algorithm for three decreasing tolerance values ($\epsilon= 9/8, 9/16$ and $9/32$) is shown for a few test cases (unsubdivided, subdivided with $\lphi$ and subdivided with $\aphi$), with 0th- and 1st-order rank invariants.} \label{table:results}
\end{center}
\end{table}

We can see that while the matching distance computed using the axis-wise linear interpolant is, for every tolerance level, equal to the matching distance between the original models, the matching distance computed using the linear interpolant can be quite different, and this both using 0th- and 1st-order persistent homology. However, this phenomenon is only seen when computing the 2D matching distance: the 1D matching distances (the numbers in the first two rows of each table) are always the same. Given that in our context topological aliasing can only be observed when using multi-measuring functions, this follows our expectations.

\subsection{Application to model precision concerns}\label{subsect:application}

When computing the matching distance between the multidimensional persistence diagrams of two models, the computation time can easily become prohibitive if the simplicial complexes representing the models are very large. For this reason, using coarser representations may be necessary. However, doing so comes at a cost in terms of accuracy. Nevertheless, using the stability property (S) of subsection~\ref{sec:stability-property} with the axis-wise interpolation, we can estimate or bound the error caused by coarsening the model, and also calculate the model precision required to reach a given error threshold. This step can be done once for every dataset, using statistical tools to obtain the expected required model precision.

We demonstrate this method using the following two datasets.

\begin{ex}[Circle]\label{ex:circle}
{\em Let $f_i:S^1 \to \R^2$, $i=1,\ldots,10^5$ be a set of random functions on the circle obtained in the following way: if $S^1$ is parametrized by the functions $x=\cos t, y=\sin t$, $t\in [0, 2\pi)$, then
\[\tilde{f}_i(x,y) = \left(\sum_{i=1}^6 \left(\alpha_{i,1}\cos(i\,t) + \beta_{i,1}\sin(i\,t)\right),\sum_{i=1}^6 \left(\alpha_{i,1}\cos(i\,t) + \beta_{i,2}\sin(i\,t)\right)\right),\]
where the $\alpha$'s and $\beta$'s are pseudo-random numbers uniformly distributed in $[-1,1)$. The function $f_i$ is then obtained from $\tilde{f}_i$ by normalizing it so that both its components take $0$ and $1$ as minimum and maximum. It can be plotted as a (not necessarily simple) closed curve in $[0,1]^2$ touching the four sides of this square. For $N=2,3,\ldots,9$, we obtain $\varphi_{i,N}$ by sampling $f_i$ at the $2^{N}$-th roots of unity on $S^1$. On the triangulation $K_N$ where $1$-simplices join successive $2^{N}$-th roots of unity, $\varphi^\urcorner_{i,N}$ can be computed, and property (S) and Theorem~\ref{th:cont-simpl-rank} guarantee that $\|\varphi^\urcorner_{i,N}-f_i\|_\infty$ is an upper bound for the matching distance $\dmm(\rho_{\varphi_{i,N}},\rho_{f_i})$. The following table shows the average and standard deviation of this upper bound over the dataset of $10^5$ functions. Another line shows the sum of mean and standard deviation, which in our tests appears to be an effective upper bound in $82$ to $86\%$ of cases. For normally distributed data, it would be such an effective upper bound in nearly $85\%$ of cases. Figure~\ref{fig:circle} plots this effective error bound in function of the percentage of original simplices kept.
\begin{center}
\begin{tabular}{|c||c|c|c|c|c|c|c|c|}
\hline
N & 2 & 3 & 4 & 5\\
\hline
$\mu$ & 0.748324 & 0.592222 & 0.346717 & 0.193834\\
\hline
$\sigma$ & 0.117896 & 0.104408 & 0.067729 & 0.032303\\
\hline
$\mu+\sigma$ & 0.866220 & 0.696630 & 0.414446 & 0.226138\\
\hline\hline
N & 6 & 7 & 8 & 9\\
\hline
$\mu$ & 0.101172 & 0.051487 & 0.025936 & 0.013012\\
\hline
$\sigma$ & 0.015632 & 0.007552 & 0.003679 & 0.001812\\
\hline
$\mu+\sigma$ & 0.116804 & 0.059039 & 0.0296155 & 0.014824\\
\hline
\end{tabular}
\end{center}

\begin{figure}[ht]
\begin{center}
\psfrag{a}{$0.39$}\psfrag{b}{$0.78$}\psfrag{c}{$1.56$}\psfrag{d}{$3.12$}\psfrag{e}{$6.25$}\psfrag{f}{$12.5$}\psfrag{g}{$25$}\psfrag{h}{$50$}\psfrag{L}[c][c][1][180]{$\mu+\sigma$}
\psfrag{0}{$0$}\psfrag{1}{$0.1$}\psfrag{2}{$0.2$}\psfrag{3}{$0.3$}\psfrag{4}{$0.4$}\psfrag{5}{$0.5$}\psfrag{6}{$0.6$}\psfrag{7}{$0.7$}\psfrag{8}{$0.8$}\psfrag{9}{$0.9$}\psfrag{x}{$1$}
\includegraphics[width=\textwidth]{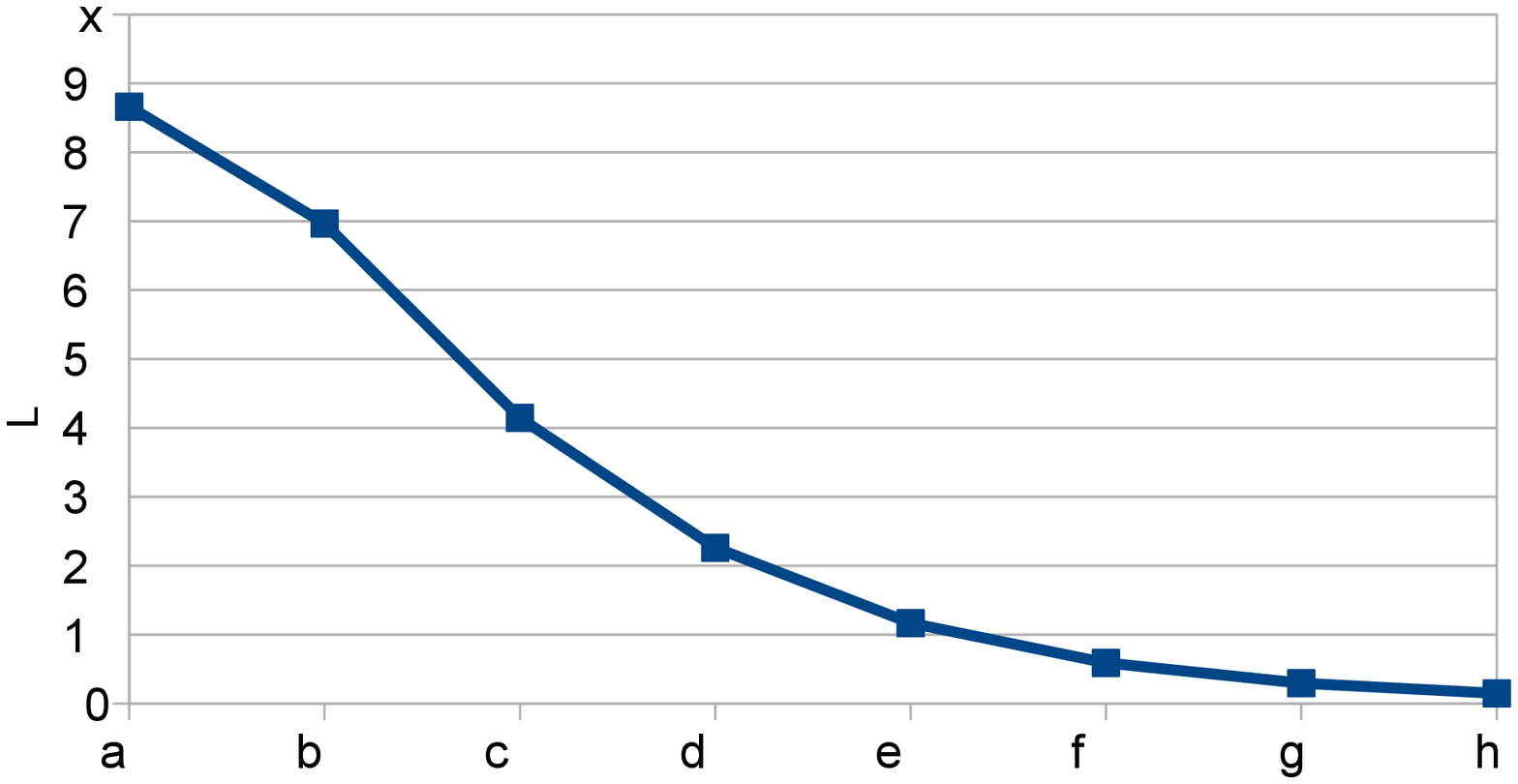}
\end{center}
\caption{Plot of $\mu+\sigma$ in function of the percentage of simplices kept from the original model, Example~\ref{ex:circle}. Horizontal axis is logarithmic.} \label{fig:circle}
\end{figure}}
\end{ex}

\begin{ex}[Torus]\label{ex:torus}
{\em Let $f_i:T \to \R^2$, $i=1,\ldots, 5000$ be a set of random functions on the torus obtained thusly: if a torus is parametrized by the equations $x = (2+1/2\,\cos t)\,\cos u, y = (2+1/2\,\cos t)\,\sin u, z = 1/2\,\sin t$, $t,u\in[0,2\pi)$, then
\[\begin{array}{lll}
\tilde{f}_i(x,y,z) =\\
\Big(\sum_{i=1}^6 \left(\alpha_{i,1}\cos(it) + \beta_{i,1}\sin(it)\right)(2+1/2\,\sum_{j=1}^6 \left(\gamma_{j,1}\cos(ju) + \delta_{j,1}\sin(ju)\right)),\\
\sum_{i=1}^6 (\alpha_{i,2}\cos(it) + \beta_{i,2}\sin(it))(2+1/2\,\sum_{j=1}^6 (\gamma_{j,2}\cos(ju) + \delta_{j,2}\sin(ju)))\Big),
\end{array}\]
where the $\alpha$'s, $\beta$'s, $\gamma$'s and $\delta$'s are pseudo-random numbers uniformly distributed in $[-1,1)$. Here again, $f_i$ is then obtained from $\tilde{f}_i$ by normalizing it so that both its components take $0$ and $1$ as minimum and maximum. Given $N=4,5,\ldots,9$, we build a uniform triangulation of the space $[0,2\pi]^2$ with the identifications $0\sim 2\pi$ in both variables $t$ and $u$, where vertices are the points $(t_i,u_j)$ with $t_i = 2\pi\,i/2^N$ and $u_j = 2\pi\,j/2^{N-2}$, and where triangles have as vertices the points $(t_i,u_j)$, $(t_{(i+1)\mod 2^N},u_j)$ and $(t_{(i+1)\mod 2^N},u_{(j+1)\mod 2^{N-2}})$, or $(t_i,u_j)$, $(t_i,u_{(j+1)\mod 2^{N-2}})$ and $(t_{(i+1)\mod 2^N},u_{(j+1)\mod 2^{N-2}})$, $i=0,\ldots,2^{N}-1$, $j=0,\ldots,2^{N-2}-1$. This space being homeomorphic to $T$, this triangulation corresponds to a triangulation of the torus. Sampling $f_i$ at the vertices of the triangulation, we obtain the function $\varphi_{i,N}$, on which $\varphi^\urcorner_{i,N}$ can be computed. As in the previous example, the following table shows the average and standard deviation, as well as their sum, of $\|\varphi^\urcorner_{i,N}-f_i\|_\infty$ over the dataset of $5000$ functions. Figure~\ref{fig:torus} further plots $\mu+\sigma$ in function of the percentage of original simplices kept.
\begin{center}
{\small\begin{tabular}{|c||c|c|c|c|c|c|c|c|}
\hline
N & 4 & 5 & 6 & 7 & 8 & 9\\
\hline
$\mu$ & 0.384139 & 0.299501 & 0.178587 & 0.097746 & 0.050357 & 0.025411\\
\hline
$\sigma$ & 0.060352 & 0.054141 & 0.033503 & 0.017958 & 0.009268 & 0.004672\\
\hline
$\mu+\sigma$ & 0.444491 & 0.353643 & 0.212090 & 0.115704 & 0.059625 & 0.030083\\
\hline
\end{tabular}}
\end{center}

\begin{figure}[ht]
\begin{center}
\psfrag{a}{$0.024$}\psfrag{b}{$0.098$}\psfrag{c}{$0.39$}\psfrag{d}{$1.56$}\psfrag{e}{$6.25$}\psfrag{f}{$25$}\psfrag{L}[c][c][1][180]{$\mu+\sigma$}
\psfrag{0}{$0$}\psfrag{1}{$0.05$}\psfrag{2}{$0.1$}\psfrag{3}{$0.15$}\psfrag{4}{$0.2$}\psfrag{5}{$0.25$}\psfrag{6}{$0.3$}\psfrag{7}{$0.35$}\psfrag{8}{$0.4$}\psfrag{9}{$0.45$}\psfrag{x}{$0.5$}
\includegraphics[width=\textwidth]{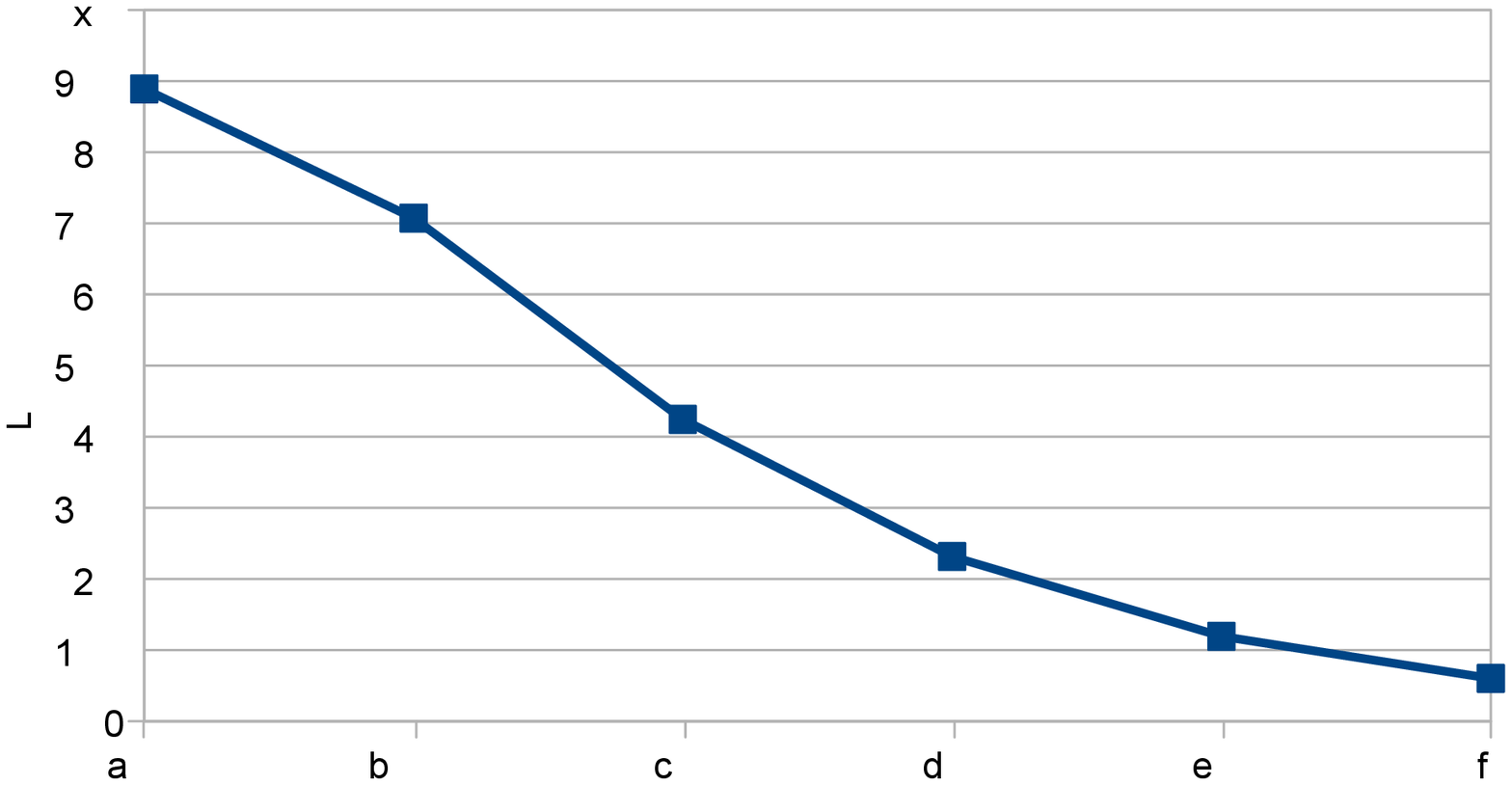}
\end{center}
\caption{Plot of $\mu+\sigma$ in function of the percentage of simplices kept from the original model, Example~\ref{ex:torus}. Horizontal axis is logarithmic.} \label{fig:torus}
\end{figure}}
\end{ex}

 \begin{center}%
 {\bfseries Acknowledgments\vspace{-.5em}}%
 \end{center}%
This work was partially supported by the following institutions: CRM-FQRNT (M.E.), University of Bologna under Marco Polo grant (N.C.), University of Modena and Reggio Emilia under Visiting Professor 2010 grant (T.K.), NSERC Canada Discovery Grant (T.K.), Fields Institute (P.F. and C.L.).  

\bibliographystyle{abbrv}
\bibliography{biblio}

\medskip

\noindent ARCES\\
Universit\`a di Bologna\\
via Toffano 2/2\\
40135 Bologna, Italia\\
\{cavazza,frosini\}@dm.unibo.it
\\~\\
D\'epartement de math\'ematiques\\
Universit\'e de Sherbrooke,\\
Sherbrooke (Qu\'ebec), Canada J1K 2R1\\
\{marc.ethier, t.kaczynski\}@usherbrooke.ca
\\~\\
\noindent Dipartimento di Scienze e Metodi dell'Ingegneria\\
Universit\`a di Modena e Reggio Emilia\\
via Amendola 2, Pad. Morselli\\
42100 Reggio Emilia
\\
clandi@unimore.it

\end{document}